\documentclass[onecolumn,a4paper,11pt]{article}

\usepackage[utf8]{inputenc} 
\usepackage[normalem]{ulem}
\usepackage{afterpage}

\usepackage{centernot}
\usepackage{amsmath,amsthm,stmaryrd}
\usepackage{bm}  
\usepackage{bbm}
\usepackage{colonequals}

\usepackage[T1]{fontenc}
\usepackage[bitstream-charter]{mathdesign}

\usepackage[margin=2cm,a4paper]{geometry}
\usepackage{verbatim}
\usepackage{xcolor}
\usepackage{units}
\usepackage{subcaption}

\usepackage[backend=bibtex,sorting=none,style=phys,biblabel=brackets,backref=true,bibencoding=utf8,maxnames=20]{biblatex}
\DefineBibliographyStrings{english}{%
  backrefpage = {page},
  backrefpages = {pages},
}
\usepackage[bookmarks,colorlinks,breaklinks]{hyperref}  
\definecolor{dullmagenta}{rgb}{0.4,0,0.4}   
\definecolor{darkblue}{rgb}{0,0,0.4}
\hypersetup{linkcolor=red,citecolor=blue,filecolor=dullmagenta,urlcolor=darkblue,pdfcreator={Me},pdfproducer={haha}} 

\usepackage{braket}
\newcommand{\proj}[1]{|#1\rangle\langle #1|}
\def\tr{{\rm Tr}}

\def\eps{\varepsilon}
\renewcommand{\phi}{\varphi}
\def\id{\mathbbm{1}}

\newtheorem{theorem}{Theorem}
\newtheorem{lemma}{Lemma}

\newtheorem{proposition}{Proposition}
\newtheorem{corollary}{Corollary}

\newtheorem{definition}{Definition}

\usepackage{mathtools}
\def\defeq{\coloneqq}

\usepackage{enumerate}

\begin{document}

\title{{Relative submajorization and its use in quantum resource theories}}
\author{{\Large Joseph M. Renes}\\
{\small Institute for Theoretical Physics, ETH Z\"urich, 8093 Z\"urich, Switzerland}
}

\maketitle
\vspace{-15mm}
\begin{abstract}
We introduce and study a generalization of majorization called relative submajorization and show that it has many applications to the resource theories of thermodynamics, bipartite entanglement, and quantum coherence.
Relative majorization is an ordering on pairs of vectors induced by stochastic transformations: One pair of vectors relatively majorizes another when there exists a stochastic transformation taking the former to the latter. Relative submajorization is a weakened version in which a substochastic matrix need only generate a vector pair obeying certain positivity conditions relative to the input pair. 
In the context of resource theories, we show that relative submajorization characterizes both the probability and approximation error that can be obtained when transforming one resource to another, also when assisted by additional standard resources such as useful work or maximally-entangled states. 
These characterizations have a geometric formulation as the ratios or differences, respectively, between the Lorenz curves associated with the input and output resources, making them efficient to compute.
We also find several interesting bounds on the reversibility of a given transformation in terms of the properties of the forward transformation. 
The main technical tool used to establish these results is linear programming duality, which is used to show that any instance of relative submajorization can be ``dilated'' to an instance of strict relative majorization.
\end{abstract}

\section{Introduction}

Majorization is an incredibly useful ordering relation on vectors, with applications to a wide variety of subjects, including quantum resource theories. 
Most famously, Nielsen showed that one bipartite pure state can be transformed into another by means of local operations and classical communication precisely when the Schmidt coefficients of the latter majorize those of the former~\cite{nielsen_conditions_1999}. 
Recently, a similar statement has been shown for transformations of pure single-system states by ``incoherent'' operations~\cite{du_conditions_2015} in the resource theory of quantum coherence.
And a relative version of majorization has been shown to characterize the possibility of resource conversion in the formulation of thermodynamics as a resource theory~\cite{janzing_thermodynamic_2000}. Here majorization is relative to the Gibbs state, rather than the uniform distribution as in usual majorization.

In all of these applications the focus is on \emph{exact} transformations: Given the ostensible input and output states, the respective majorization conditions determine whether the transformation is possible or not. 
However, in the operationally-motivated setting of resource theories, approximate and probabilistic transformations are more relevant, as are transformations aided by additional ``standard'' resources such as maximally-entangled states or useful work, or even combinations of all of these.
Instead of an all-or-nothing proposition, the focus now shifts to quantitative questions such as how much error will be incurred in a particular transformation or how many additional resources will be needed to complete it.
These issues have also be addressed in the framework of majorization. 
For thermodynamics, a majorization characterization of work-assisted transformations was given in~\cite{horodecki_fundamental_2013} and of probabilistic transformations in~\cite{alhambra_what_2015}. Transformations of pure bipartite states assisted by maximally-entangled states were characterized in~\cite{egloff_measure_2015}. 

In this paper we show that for resource-assisted transformations, both the achievable probability and the achievable approximation error are precisely determined by an extension to majorization we term \emph{relative submajorization}.
This includes and substantially extends the previous results. Approximate resource-assisted transformations have not been so characterized before, to our knowledge, and the combination of probabilistic and resource-assisted transformations is also novel. 
Majorization has an elegant geometric formulation involving the ordering of the Lorenz curves of the two vectors, and we also show that this formulation extends to relative submajorization. 
Then we can show, in particular, that approximate or probabilistic resource-assisted transformations can be directly characterized in terms of the ratios or differences, respectively, of the Lorenz curves of the resources themselves. 
Figures~\ref{fig:aj} and \ref{fig:joint} depict our results as applied to thermodynamics, and the captions reference the precise technical statements.

\begin{figure}[!p]
\captionsetup[subfigure]{position=b}
\centering
\subcaptionbox{Feasibility of $\lambda_0,z_0$ via Lorenz curves}[0.48\textwidth]{\includegraphics{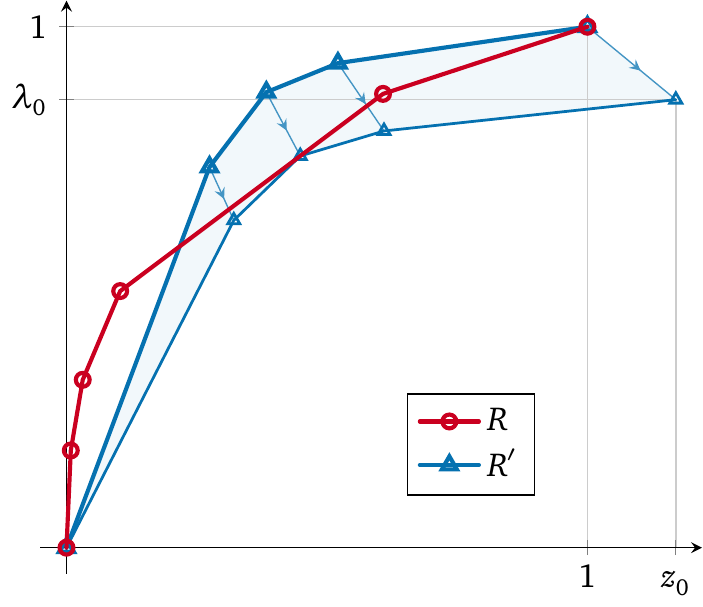}}
\hspace{3mm}
\subcaptionbox{Feasibility region}[0.48\textwidth]{\includegraphics{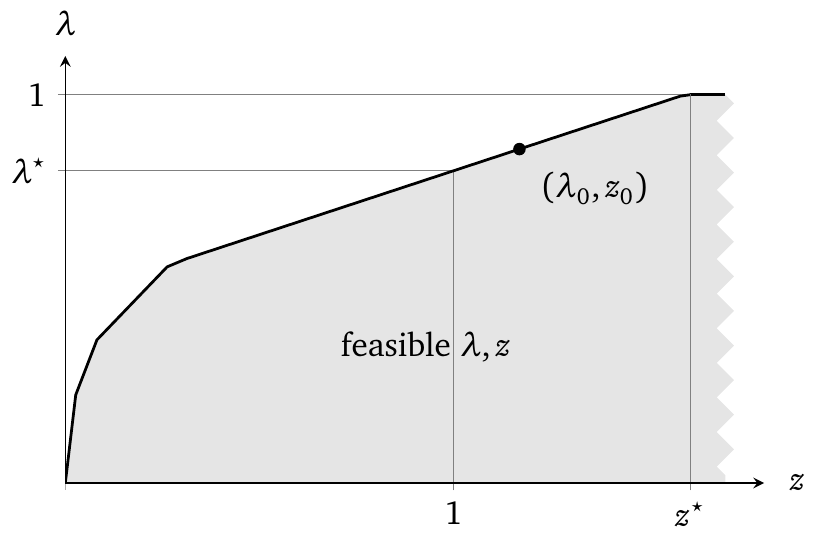}}
\caption{Feasible transformation probabilities and work-assistance for converting resource $R$ to $R'$ by thermal operations at a fixed (inverse) temperature $\beta$.
In (a), the transformation $R\to R'$ is not possible deterministically, as the Lorenz curve of $R$ (defined in \S\ref{sec:testing}) does not lie everywhere above that of $R'$. However, by supplementing the input with additional work, the success probability of the transformation can be increased. The additional work $W$ corresponds to horizontally scaling the Lorenz curve of the output by $z_0=e^{\beta W}$, and as shown in \S\ref{sec:probworktrans}, $\lambda_0$ is an achievable transformation probability if rescaling the output vertically by $\lambda$ produces a curve below that of the input $R$.
By \eqref{eq:lambdaofz} (cf.\ Corollary~\ref{cor:boundary}), for fixed $z$ the optimal $\lambda$ is the smallest ratio of the input Lorenz curve to that of the horizontally-scaled output.  
The set of all feasible $(\lambda,z)$ pairs is depicted in (b), with the particular $(\lambda_0,z_0)$ from (a) indicated, as well as the optimal probability $\lambda^\star$ involving no work and the optimal amount of work needed for deterministic transformation.
}
\label{fig:aj}
\end{figure}

\begin{figure}[h!]
\captionsetup[subfigure]{position=b}
\centering
\subcaptionbox{
    The differences between two Lorenz curves are related to the two kinds of error in transforming $R\to R'$, defined in \S\ref{sec:approx}. 
    The dashed line shows the maximum vertical difference between $R'$ and $R$, the dotted line the maximum horizontal difference between $R$ and $R'$ (note the different orderings). 
    The former is equal to the first kind of error $\eps(R\to R')$, while the latter is directly related 
    to the second kind of error $\eta(R\to R')$ (cf.\ Prop.~\ref{prop:approxbounds}).}[0.48\textwidth]{\includegraphics{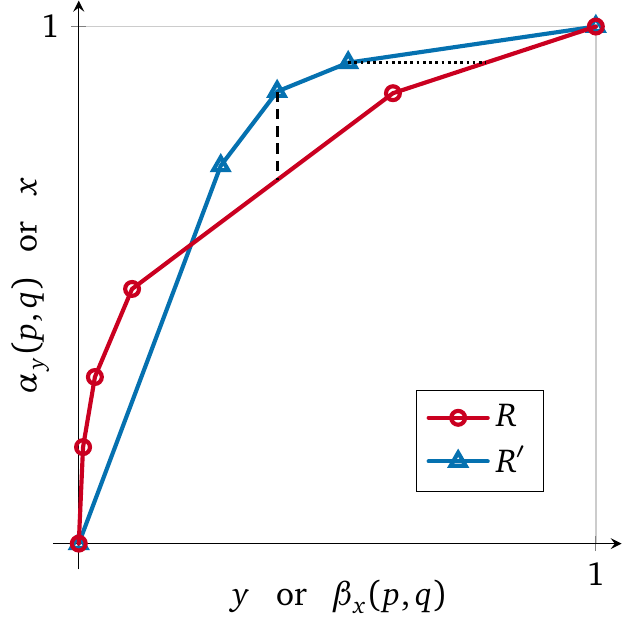}}
\hspace{3mm}
\subcaptionbox{
    The value of the Lorenz curve for resource $R$ at $z$ is the maximum probability $\lambda^\star_z(R\to 1)$ that an amount of work $W$ can be obtained from $R$ by thermal operations at inverse temperature $\beta$, with $z=e^{-\beta W}$ (cf.\ Prop.~\ref{prop:workvalue}). 
    Meanwhile, the tangent with slope $z'$ intercepts the vertical axis at the lowest approximation error $\eps^\star_{z'}(1\to R)$ for producing $R$ and extracting work $W'$, for $z'=e^{-\beta W'}$ (cf.\ Prop.~\ref{prop:workcost}). 
    For $z'>1$, as here, work is expended to create an approximation to $R$.}[0.48\textwidth]{\includegraphics{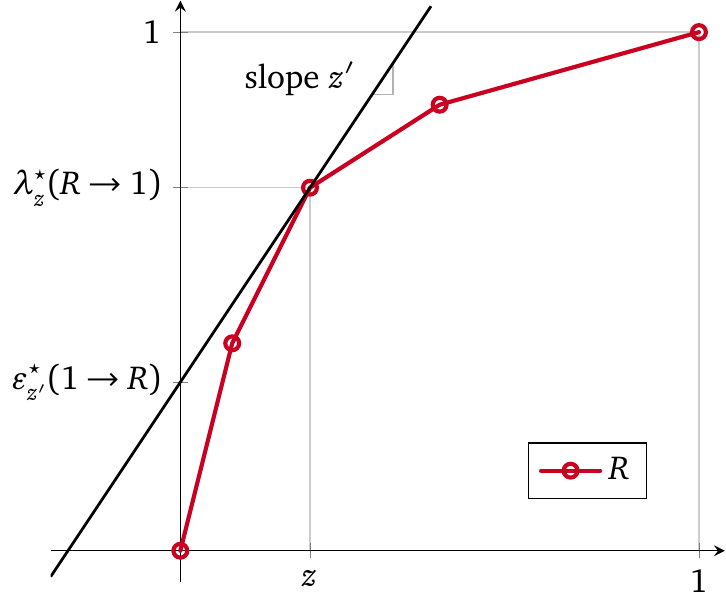}\par\vspace{2.8mm}}
\caption{Thermodynamical interpretation of (a) the maximum differences between Lorenz curves of two resources and (b) a resource's Lorenz curve and its Legendre transform. 
}
\label{fig:joint}
\end{figure}

\afterpage{\clearpage}

The mathematical machinery of relative submajorization also leads to several interesting bounds relating to the reversibility of a transformation, generalizing results obtained in~\cite{egloff_measure_2015,alhambra_what_2015,wehner_work_2015}. 
In particular, we find a lower bound on the work required to probabilistically reverse a transformation in terms of the work cost of the forward transformation and the probabilities involved (Prop.~\ref{prop:reversibility}), as well as an upper bound on the approximation error given the change in relative entropy of the resource (Prop.~\ref{prop:epsrecoverability}). 
Applied to the resource theory of coherence (and similar results hold for bipartite entanglement), we determine the coherence cost of transforming any given pure state into another using only incoherent operations (Prop.~\ref{prop:coherencecost}) and find bounds relating the transformation fidelity to the coherence cost of the transformation, as well as to the entropy of the squared amplitudes of the states in the classical basis (Prop.~\ref{prop:coherencebounds}). 


The remainder of the paper is structured as follows. 
We first define relative submajorization in \S\ref{sec:mathtools} and study its properties in a purely mathematical setting, without consideration of any eventual applications. 
Then, in \S\ref{sec:resourcethermo} we apply the framework to the resource theory of thermodynamics, in particular for quasiclassical resources that are diagonal in the energy eigenbasis. 
In \S\ref{sec:resourcecoherence} we consider the resource theory of pure bipartite entanglement and the closely related resource theory of pure state quantum coherence. 
We conclude with a brief summary in \S\ref{sec:summary}. 

\section{Mathematical Framework}
\label{sec:mathtools}
\subsection{Notation and Conventions}
Let us first establish notation and our conventions for vectors and matrices. 
We will be concerned with real-valued vectors, elements of $\mathbb R^n$, especially those in the positive cone $\mathbb R^n_+$ (those with non-negative entries). 
Vectors are considered to be column vectors unless otherwise specified, and $1_n$ denotes the length-$n$ vector of all ones. 
The positive part of a vector $v$ is denoted $(v)_+$, i.e.\ $v$ with all negative entries set to zero.
The expression $|v|$ indicates the one-norm of $v$, i.e.\ the sum of the components for vectors in the positive cone. 
The inequality $v\geq 0$ is taken to mean $v\in \mathbb R_+^n$. 
The simplex of valid probability distributions (probability mass functions) on $n$ outcomes, i.e.\ $v\in \mathbb R_+^n$ with $|v|=1$, is denoted $S_n$. 
The symbol $\oplus$ denotes the direct sum of vectors, and $\otimes$ the direct product.

A particularly useful distance measure on $S_n$ is the variational distance $\delta(p,q):=\max\{ t\cdot (p-q):0\leq t\leq 1_n\}=\min\{|v|:v\geq p-q\}$, as it is directly related to the probability that any experiment could detect the difference between a system described by $p$ or by $q$.
Meanwhile, the relative entropy of $p$ to $q$, denoted $D(p,q)$ is given by $D(p,q)=\sum_{k=1}^n p_k \log p_k/q_k$.
Here, and the remainder of the paper, we use the natural logarithm. 

We will also be concerned with real-valued matrices, and denote by $\mathbb R^{n\times m}_+$ the set of real-valued $n\times m$ matrices with nonegative entries; for a given matrix $M$, this condition is also expressed as $M\geq 0$. 
An $m\times n$ matrix $M$ with $M\geq 0$ is \emph{stochastic} when $1_m^TM =1_n^T$, and \emph{substochastic} when $1_m^TM \leq 1_n^T$. 
It is \emph{doubly stochastic} when, in addition to being stochastic, $n=m$ and $M 1_n=1_m$.

Linear programming and in particular its duality is the main mathematical tool behind the results we shall present. 
A formulation useful for our purposes here is given in Appendix~\ref{sec:lp}.

\subsection{Testing regions and Lorenz curves}
\label{sec:testing}
We begin by defining the \emph{testing region} of two (unnormalized) probability distributions.
This is closely related with Neyman-Pearson hypothesis testing, and will turn out to be very closely related to majorization. 

Vectors $0\leq t\leq 1_n$ are ``tests'' because the vectors $t$ and $1_n-t$ can be associated with the outcomes of a dichotomous measurement, such that for underlying probability distribution $p$, the probabilities of the outcomes are just $t\cdot p$ and $1-t\cdot p$, respectively. 
\begin{definition}
For two positive vectors $p,q\in\mathbb R_+^n$, the \emph{testing region} $T(p,q)$ is given by
\begin{align}
T(p,q) \defeq \{(x,y)\in\mathbb{R}^2_+:(x,y)=(t\cdot p,t\cdot q), 0\leq t\leq 1_n\}.
\end{align}
\end{definition}
The testing region is a convex set whose extreme points correspond to the extreme points of the set of tests $t$, those whose entries are either zero or one. 
Hence it contains the points $(0,0)$ and $(|p|,|q|)$, corresponding to the tests $t=0$ and $t=1_n$, respectively, as well as the line joining them.
Furthermore, it is symmetric under $(x,y)\mapsto (|p|-x,|q|-y)$, corresponding to $t\to 1_n-t$.
An example is depicted in Figure~\ref{fig:testingregion}.

\begin{figure}[ht]
{\centering 
\includegraphics{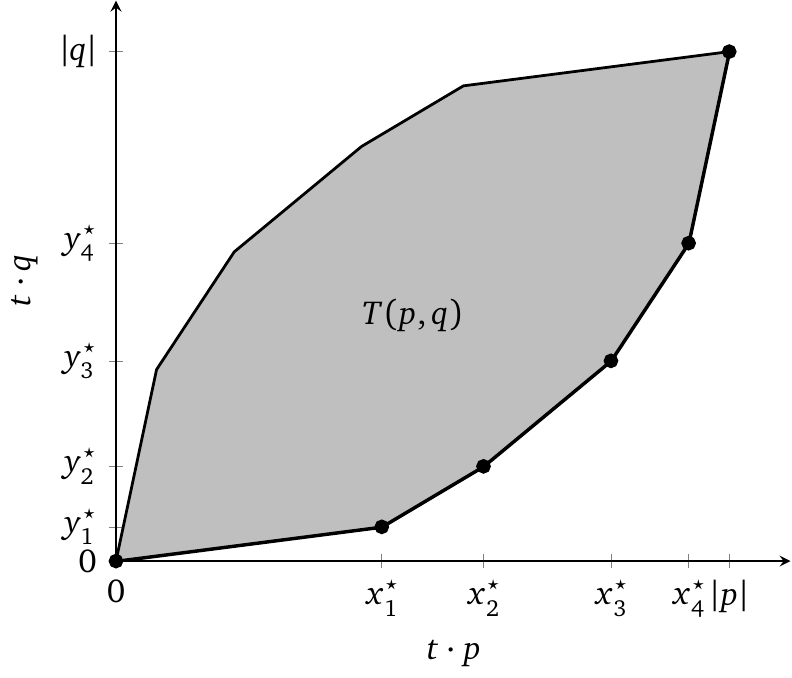}

}
\caption{
  \label{fig:testingregion}
  Example of a testing region $T(p,q)$ for $p,q\in \mathbb R^n_+$ with $n=5$. The region is a subset of the positive quadrant of $\mathbb R^2$ consisting of all pairs $(t\cdot p,t\cdot q)$ for test vectors $0\leq t\leq 1_n$. 
  The testing region is convex and due to the symmetry under $t\mapsto 1_n-t$, meaning $(|p|-x,|q|-y)\in T(p,q)$ for $(x,y)\in T(p,q)$, it is sufficient to specify the region by describing its lower (or upper) boundary. 
  The $y$ coordinates of the boundary are the function values $\beta_x(p,q)$, while the  while the $x$ coordinates are the function values $\alpha_y(p,q)$.
  In fact, $T(p,q)$ is the convex hull of just $2n$ extreme points, as shown in Lemma~\ref{lem:Tcornerpoints}. 
  Those along the lower boundary are indicated, the $x$ coordinates of which comprise the set $T_x^\star(p,q)$ and the $y$ coordinates $T_y^\star(p,q)$. 
}
\end{figure}
\emph{A priori}, there could be as many as $2^n$ extreme points of $T(p,q)$, one for every extreme point in the set of tests.
However, the following lemma shows that there are at most $2n$ extreme points for any $T(p,q)$. 
\begin{lemma}
\label{lem:Tcornerpoints}
For any $p,q\in \mathbb R_+^n$, the extreme points of the lower boundary of $T(p,q)$ are, along with (0,0),  
\begin{align}
(x^\star_k,y^\star_k)=\left(\sum_{j=1}^k p_{\pi(j)} ,\sum_{j=1}^k q_{\pi(j)} \right),
\end{align}
where $k=1,\dots, n$ and $\pi$ is any permutation of $\{1,\dots n\}$ such that the sequence $(p_{\pi(j)}/q_{\pi(j)})_{j=1}^n$ is nonincreasing. 
\end{lemma}
Intuitively, these are the points ``farthest away'' from the diagonal running from $(0,0)$ to $(|p|,|q|)$. 
For convenience later, define $T_x^\star(p,q)$ to be the set $\{x^\star_k\}_{k=1}^n$ and similarly $T_y^\star(p,q)=\{y_k^\star\}_{k=1}^n$. 
These points will be referred to as ``elbows'' of the boundary curve.

By symmetry, it is sufficient to describe $T(p,q)$ by its lower (or upper) boundary, and to prove the statement, it is useful to work with the function taking the $x$ coordinate of points on the lower boundary to their corresponding $y$ coordinates. 
Setting
\begin{align}
\label{eq:betadef}
\beta_x(p,q)
\defeq \min \,\{y:(x,y)\in T(p,q)\},
\end{align}
$x\mapsto \beta_x(p,q)$ is an increasing, convex function with domain $[0,|p|]$. 
By construction, the value $\beta_x(p,q)$ is the solution of a linear program, one which satsifies strong duality. 
The dual optimizations are as follows.
\begin{subequations}
\begin{align}
\label{eq:betaprimal}
  &\begin{array}{r@{\,\,}rl}
     \beta_x(p,q)=& \underset{t}{\rm infimum} & t\cdot q\\ 
     &\text{\rm subject to} & t\cdot p\geq x,\\ 
     && 0\leq t\leq 1_n.
     \end{array}\\
     \label{eq:betadual}
  &\begin{array}{r@{\,\,}rl}
  \phantom{\beta_x(p,q)}=& \underset{\mu,s}{\rm supremum} & \mu x-1_n\cdot s\\
  &  \text{\rm subject to} & \mu p-q\leq s,\\
  && \mu,s\geq 0.
  \end{array}
\end{align}
\end{subequations}
Strong duality is assured by the existence of a feasible test for the dual, namely $\mu=s=0$.
Nominally, the inequality constraint in the primal form should be an equality, to match with the definition in \eqref{eq:betadef}. 
However, the test optimal in the primal has the property that $t\cdot p=x$, since any $t$ with $t\cdot p>x$ can simply be rescaled to meet the constraint with equality and thereby decrease the objective function (equality can also be seen from the complementary slackness conditions). 
Moreover, the primal form \eqref{eq:betaprimal} should now be taken as the definition of the function $x\mapsto \beta_x(p,q)$, since we no longer need to restrict the domain of $x$.
For $x<0$, the function simply takes the value zero ($t=0$ in the primal; $\mu=s=0$ in the dual), and for $x>|p|$ the value $+\infty$ as the primal problem is no longer feasible. 

Using the dual, we can easily prove Lemma~\ref{lem:Tcornerpoints}. 
\begin{proof}[Proof of Lemma~\ref{lem:Tcornerpoints}]
Letting $t_k$ be the vector whose components are 1 for indices $\pi(1),\dots,\pi(k)$ and zero otherwise. 
Then, by the primal formulation, direct calculation gives $\beta_{x_k^\star}(p,q)\leq y_k^\star$ for all $k$. By mixing the $t_k$ tests, it immediately follows that $\beta_x(p,q)\leq \lambda y_k^\star+(1-\lambda)y_{k+1}^\star$ for $x=\lambda x_k^\star+(1-\lambda)x_{k+1}^\star$ with any $\lambda \in [0,1]$.  

To prove the opposite inequalities, set $\mu_k$ such that $\mu_kp_{\pi(k+1)}=q_{\pi(k+1)}$ and then choose $s_k=(\mu_k p-q)_+$ in the dual. By construction,
\begin{align}
1_n\cdot s_k=\sum_{j=1}^k \mu_k p_{\pi(j)}-q_{\pi(j)}=\mu_k x_k^\star-y_k^\star.
\label{eq:dualbetafeasible}
\end{align}
Thus, $\beta_{x_k^\star}(p,q)\geq y_k^\star$ for all $k$. For arbitrary $x=\lambda x_k^\star+(1-\lambda)x_{k+1}^\star$, note that $\lambda x_k^\star+(1-\lambda)x_{k+1}^\star=x_k^\star+(1-\lambda)p_{\pi(k+1)}$ and the analogous statement holds for $y_k^\star$. Thus, 
\begin{align}
\beta_x(p,q)
&\geq y_k^\star+(1-\lambda)\mu_k p_{\pi(k+1)}\\
&=y_k^\star+(1-\lambda)q_{\pi(k+1)},
\end{align}
completing the proof.
\end{proof}

Later the following lemma will also be useful.
\begin{lemma}
\label{lem:betalemma}
For any $p,q\in \mathbb R_+^n$ and $r\in\mathbb R_+^m$, we have, for all $x\in \mathbb R$,
\begin{align}
\beta_x(p\oplus 0_m,q\oplus r)&=\beta_x(p,q),\quad \text{and}\label{eq:betasum}\\
\tfrac1{|r|}\beta_{x|r|}(p\otimes r,q\otimes r)&=\beta_x(p,q).\label{eq:betaproduct}
\end{align}
\end{lemma}
\begin{proof}
The optimal test for the lefthand quantity in \eqref{eq:betasum} may as well have the form $t\oplus 0_m$, since feasibility is not affected by having positive $t'$ in the second system and the objective function will only become larger. Then the optimal test for either quantity is a feasible test for the other.

The lefthand side of \eqref{eq:betaproduct} is clearly no larger than the righthand side, since any test $t$ for the latter can be used in the former via $t\otimes 1_m$.
Conversely, the optimal $s$ in the dual of $\beta_x(p,q)$ is clearly $(\mu p-q)_+$, from which it follows that the left side must be larger than the right. 
\end{proof}


In describing the lower boundary, we could just as well consider the map from the $y$- to the $x$-coordinates:
\begin{align}
\alpha_y(p,q)\defeq \max\,\{x:(x,y)\in T(p,q)\}.
\end{align}
The function $y\mapsto \alpha_y(p,q)$ has domain  $[0,|q|]$ and is again an increasing function, but now concave. 
By construction, $\alpha(p,q)$ and $\beta(p,q)$ are 
inverses of each other, in that 
\begin{align}
\beta_{\alpha_y(r,g)}(r,g)&=y,\qquad \forall y\in [0,|q|]\\
\alpha_{\beta_x(r,g)}(r,g)&=x,\qquad \forall x\in [0,|p|].
\end{align}
When $|p|=|q|$, the symmetry of $T(p,q)$ implies that 
\begin{align}
\alpha_y(p,q)+\beta_{1-y}(q,p)=1.
\label{eq:alphabeta}
\end{align}

The linear program formulation of $\alpha_y(p,q)$ is as follows:
\begin{subequations}
\begin{align}
  &\begin{array}{r@{\,\,}rl}
     \alpha_y(p,q) = & \underset{t}{\rm supremum} &  t\cdot p\\
     &\text{\rm subject to} & t\cdot q\leq  y,\\
     && 0\leq t\leq 1_n,
     \end{array}\\
  &\begin{array}{r@{\,\,}rl}
  \phantom{\alpha_\eps(p,q)} =&  \underset{\nu,s}{\rm infimum} &  \nu y+1_n\cdot s\\
     &\text{\rm subject to} & s\geq p-\nu q,\\
     && \nu,s\geq 0.
     \end{array}
  \end{align} 
  \end{subequations}
Again in this formulation the domain of $y\mapsto \alpha_y(p,q)$ is extended to all of $\mathbb R$. 
When $y<0$, the function takes the value $-\infty$ as the program is no longer feasible. 
On the other hand, for $y>|q|$, the function takes the value $|p|$, which can be seen by choosing $\nu=0$ and $s=p$ in the dual formulation.
 



All four possible functions describing the two boundaries are used in the literature. 
The function $x\mapsto \beta_x(p,q)$ is commonly used in classical information theory, particularly for $|p|=|q|=1$, when it describes the minimal type-II error of a Neyman-Pearson asymmetric hypothesis test between $p$ and $q$ for type-I error no larger than $1-x$~\cite{polyanskiy_saddle_2013}. 
In quantum thermodynamics it is common to use $\alpha_y(p,q)$~\cite{horodecki_fundamental_2013}.
We shall adopt this use here, and refer to the graph of the function $y\mapsto \alpha_y(p,q)$ as the Lorenz curve of $(p,q)$. 
For an illustration, see Fig.~\ref{fig:LorenzCurve}. 

Meanwhile, in statistics, it is common to consider the upper boundary $\hat \alpha(p,q):y\mapsto \hat\alpha_y(p,q)=|p|-\alpha_{|q|-y}(p,q)$. This is the maximum power at significance level $y$.\footnote{However, in the statistics literature it is common to call this function $\beta$ (!)~\cite{lehmann_testing_2005,torgersen_comparison_1991}.} If we follow Harremo\"es~\cite{harremoes_new_2004} and consider $q$ to be uniform, then the Lorenz curve of economics as commonly defined today~\cite[Ch.\ 17C]{marshall_inequalities:_2009} is the $x$ coordinate of the upper boundary considered as a function of the $y$ coordinate.\footnote{Though Lorenz's original definition in~\cite{lorenz_methods_1905} was the statistician's choice, $y\mapsto\hat \alpha_y(p,q)$.}

\begin{figure}[ht]
{\centering
\includegraphics{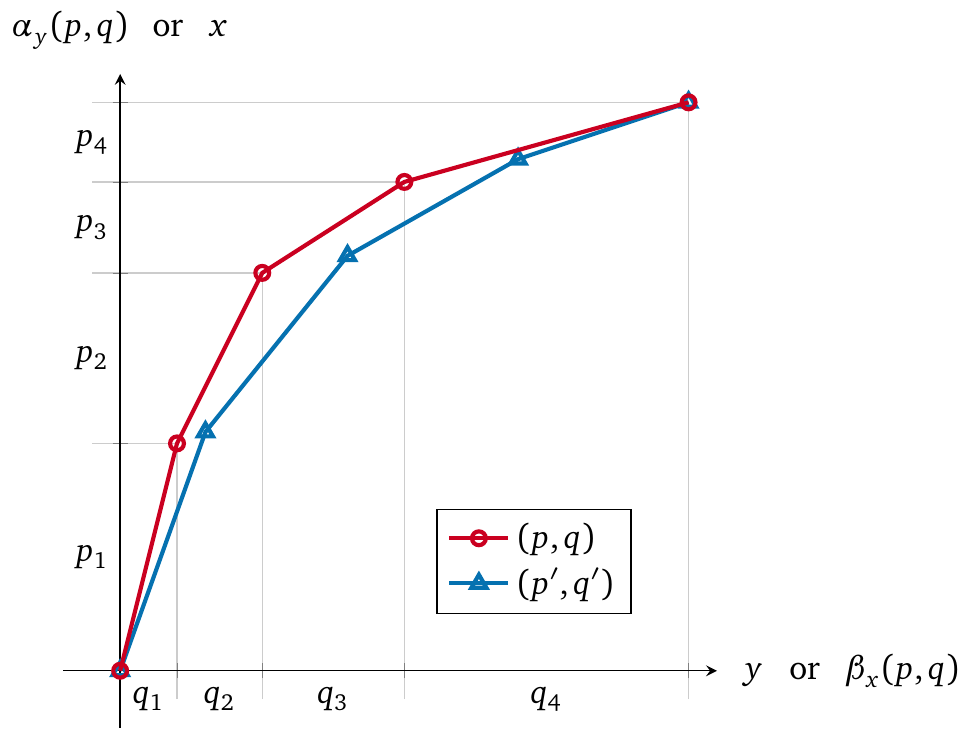}

}
\caption{\label{fig:LorenzCurve} Lorenz curves for two normalized pairs $(p,q)$ and $(p',q')$, with $(p,q)\succeq (p',q')$.
The value of the Lorenz curve for $(p,q)$ as a function of the horizontal value $y$ is simply $\alpha_y(p,q)$. 
As a function of the vertical value $x$, it is $\beta_x(p,q)$. 
The advantage of the former (standard) interpretation is that $(p,q)\succeq (p',q')$ corresponds to the Lorenz curve of the former lying above that of the latter. 
The distances between horizontal and vertical coordinate values of the elbows of the curve are simply the components of $q$ and $p$, respectively (cf.\ Lemma~\ref{lem:Tcornerpoints}). 
Here, $p$ and $q$ are already assumed to be ordered so that the ratios $p_k/q_k$ form a decreasing sequence.
Note that it suffices to check that the elbows of $(p',q')$ lie below the $(p,q)$ curve.}
\end{figure}



\subsection{Majorization}
Majorization is an incredibly useful ordering relation on vectors, with applications to a wide variety of subjects, far too many to list here. An excellent overview is given in~\cite{marshall_inequalities:_2009}.
Although there are many equivalent characterizations of majorization, to this author the most fundamental appears to be the notion that $p$ majorizes $q$ when $q$ can be obtained from $p$ by the action of a doubly stochastic matrix. 

Majorization can be usefully extended to pairs of vectors, as introduced by Veinott~\cite{veinott_least_1971}, and called $d$-majorization or relative majorization.\footnote{In~\cite{marshall_inequalities:_2009} relative majorization refers to the case when the second vector is the same for both pairs.}  
Actually, the notion is older in the context of statistics, where it arises naturally in the context of comparing experiments (see, e.g.\ the influential papers by Blackwell~\cite{blackwell_comparison_1951,blackwell_equivalent_1953} or the encyclopedic volume by Torgersen~\cite{torgersen_comparison_1991}, particularly \S9.1).
Now the criterion is that there exists a stochastic matrix $M$ that takes each vector in the input pair to the corresponding vector in the output pair. 
Here we are (mostly) concerned with vectors in the positive cone $\mathbb R^n_+$, 
though the definition can be made in more general settings; see~\cite[\S 14.B]{marshall_inequalities:_2009} or \cite{torgersen_comparison_1991}.
\begin{definition}[Relative majorization]
For $p,q\in \mathbb R^n_+$ and $p',q'\in \mathbb R^{n'}_+$, the pair $(p,q)$ is said to \emph{majorize} (or \emph{relatively majorize}) the pair $(p',q')$, denoted $(p,q)\succeq (p',q')$, when there exists an $n'\times n$ stochastic matrix $M$ such that
\begin{subequations}
\label{eq:majorizationdef}
\begin{align}
Mp&=p',\\
Mq&=q'.
\end{align}
\end{subequations}
\end{definition}
Clearly, for $(p,q)\succeq (p',q')$ to hold, we must have $|p'|=|p|$ and $|q'|=|q|$. 
Note that $(p,q)\succeq (p',q')\Leftrightarrow (q,p)\succeq (q',p')$, so there is no inherent ordering to the pair. 

Usual majorization $p\succeq q$ is just the case that $n=n'$ and the second (or first) vector in the pair is proportional to $1_n$, the vector of all ones; this ensures that the map from $p$ to $q$ is \emph{doubly} stochastic. 
For vectors of unequal length 
we may just embed both in some $\mathbb R^{m}$ with $m\geq \max\{n,n'\}$, by padding with zeros, and then define $p\succeq q$ to be $(p,1_{m})\succeq (q,1_{m})$.

An important property of relative majorization is that it is completely characterized by the inclusion of $T(p',q')$ in $T(p,q)$. 
Stated in terms of the Lorenz curve, we have 
\begin{align}
\label{eq:majorizationbeta}
(p,q)\succeq (p,q') \quad \Leftrightarrow \quad \beta_x(p,q)\leq \beta_x(p',q')\quad \forall x\in \mathbb R.
\end{align}
Proof of this statement can be traced back to results of Blackwell~\cite{blackwell_equivalent_1953} (see \cite[Ex.\ 9.1.5]{torgersen_comparison_1991}), while direct proofs in these terms were given in \cite{ruch_mixing_1978} and \cite{uhlmann_ordnungsstrukturen_1978}. 
In the second statement it is indeed enough to test the inequality only at $x\in T_x^\star(p',q')$. 
The extension to all $x\in T_x(p',q')$ follows since the function $x\mapsto \beta_x(p',q')$ interpolates linearly between $\beta_k^\star(p',q')$ and $\beta_{k+1}^\star(p',q')$ for the largest $k$ such that $\alpha_k^\star(p',q')<x$, while the graph of $x\mapsto \beta_x(p,q)$ lies below the endpoints and is convex.

\subsection{Relative submajorization}
Comparison of \eqref{eq:majorizationdef} and \eqref{eq:majorizationbeta} immediately reveals that, for general vectors, the equality conditions in the former are sufficient but not necessary for the Lorenz curves to be ordered as in the latter. 
By weakening the equality constraints and considering substochastic matrices, we can define a generalization whose Lorenz curves are still ordered. 
For reasons which will become apparent later, we term this generalization ``relative submajorization''. 

\begin{definition}[Relative submajorization]
The pair $(p,q)\in \mathbb R^{2n}_+$ is said to \emph{(relatively) submajorize} $(p',q')\in \mathbb R_+^{2n'}$, denoted $(p,q)\succ (p',q')$,
when there exists an $M\in \mathbb R_+^{n'\times n}$ such that 
\begin{subequations}
\label{eq:Lorenzmajorizationdef}
\begin{align}
Mp&\geq p',\label{eq:pcondition}\\
Mq&\leq q',\label{eq:qcondition}\\
1_{n'}^TM&\leq 1_n^T.\label{eq:Mcondition}
\end{align}
\end{subequations}
Additionally, if \eqref{eq:pcondition} holds with equality, we write $(p,q)\curlyeqsucc (p',q')$, while if \eqref{eq:qcondition} and \eqref{eq:Mcondition} hold with equality, we write $(p,q)\succcurlyeq (p',q')$.
\end{definition}
The symmetry between the two vectors in the pair is now lifted, since the inequalities run in different directions. 
The condition can only hold when $|p|\geq |p'|$, but the second member of the pair faces no such constraint.

Nevertheless, when $|p'|=|p|$ and $|q'|=|q|$ both hold, submajorization is equivalent to strict majorization. 
Clearly the latter implies the former, and the reverse implication hinges on \eqref{eq:majorizationbeta}. 
Specifically, let $t'$ be the optimal test in $\beta_x(p',q')$, so that $x=t'\cdot p'$ and $\beta_x(p',q')=t'\cdot q'$. 
Defining $t=M^Tt'$, substochasticity of $M$ ensures that $t$ is a valid test. The remaining  inequalities give $t\cdot p\geq x$ and $t\cdot q\leq \beta_x(p',q')$, and therefore $\beta_x(p,q)\leq \beta_x(p',q')$ for all $x\in [0,1]$.
Using \eqref{eq:majorizationbeta} completes the reverse implication. 

As shown in Lemma~\ref{lem:upperrm} of the Appendix, it turns out that $\succ$ and $\curlyeqsucc$ are equivalent.
Moreover, by Lemma~\ref{lem:directroute}, $\succ$ and $\succcurlyeq$ are equivalent when $|q'|=|q|$,

The former equivalence motivates the name ``submajorization'', since again setting the second vector in the pair proportional to $1_n$, we recover the usual notion of (weak) submajorization that $p$ (weakly) submajorizes $q$ when there is a doubly substochastic matrix $M$ such that $Mp=q$ (cf.\ \cite[Th.\ 2.C.4]{marshall_inequalities:_2009}).\footnote{Note, however, that this is apparently \emph{not} compatible with the notion of submajorization given by Torgersen~\cite{torgersen_comparison_1991}. In fact, what we have called submajorization is a generalization of his notion of \emph{supermajorization} (compare Theorem~\ref{thm:rsm} when $|p'|=|p|$ with condition (iii$'''$), pp.\ 577).}
Even though we can recover standard majorization by setting the first vector in the pair proportional to $1_n$, doing so here would not yield supermajorization of \cite{marshall_inequalities:_2009}, since the condition \eqref{eq:Mcondition} forces $M$ to be sub- and not super-stochastic.
While it is possible to embed or convert relative (sub)majorization into standard (sub)majorization~\cite{ruch_principle_1976,ruch_mixing_1978,joe_majorization_1990}, at which point one could attempt to glean properties of the latter from those of the former, we refrain from doing so here as the proofs of the interesting properties presented below are essentially the same in the two cases.

We constructed the inequalities in \eqref{eq:Lorenzmajorizationdef} so that the forward implication in \eqref{eq:majorizationbeta} still holds, i.e.\ so that the Lorenz curves would still be ordered.  
Importantly, \emph{the reverse implication also holds}, even though the differing norms of the vectors prevent their associated testing regions from obeying an inclusion relation. 
The ultimate reason for this is that two given pairs of vectors can be ``dilated'' to two new pairs which satisfy the testing region inclusion condition.\footnote{
The notion of dilation is that of Naimark, Stinespring, and so forth, where the general object is shown to be ``compressed'' or ``projected'' version of a nicer object on a larger space (see Paulsen~\cite{paulsen_completely_2003}, from which this description was taken), and not the notion of a dilation as a particular kind of stochastic map as used by, e.g.\  Torgersen~\cite[pp.\ 346]{torgersen_comparison_1991}.} 
Put differently, while strict majorization is a special case of submajorization, any instance of submajorization can be dilated to an instance of strict majorization in the sense that the vectors involved the submajorization condition are projected versions of the vectors involved in the strict majorization condition.
This state of affairs is encapsulated in the following theorem.
\begin{theorem}
\label{thm:rsm}
For any $p,q\in \mathbb R_+^n$ and $p',q'\in \mathbb R_+^{n'}$, the following are equivalent: 
\begin{enumerate}[(a)]
\item $(p,q)\succ (p',q')$,
\item $(p\oplus 0,q\oplus z q)\succeq (p'\oplus s',q'\oplus z^{-1}q')$ for $s'\propto q'$ with $|s'|+|p'|=|p|$ and $z=|q'|/|q|$,
\item $\beta_x(p,q)\leq \beta_x(p',q')$ for all $x\in T_x(p',q')$, indeed for all $x\in T_x^\star(p',q')$.
\end{enumerate}
\end{theorem}

The proof of the the theorem is divided into Lemmas \ref{lem:uglydual}, \ref{lem:directroute}, and  \ref{lem:elbows} in the Appendix, relying on strong duality of linear programming. 
We also sketch an alternate, more constructive proof based on interpolation of tests.

\subsection{Feasible norms in relative submajorization}

Given any two pairs of vectors $R=(p,q)$ and $R'=(p',q')$, it is interesting to consider the possible combinations of $\lambda,z\geq 0$ such that $(p,q)\succ (\lambda p',zq')$.
Since we are considering scaling each vector in the second pair by some positive constant, we might as well focus on the case $|p'|=|p|=1$ and $|q'|=|q|=1$, and we refer to the corresponding $R$ and $R'$ as normalized.
Let us denote the set of all possible combinations for normalized $R$ and $R'$ by $\mathcal S(R,R')$, that is
\begin{align}
\mathcal S(R,R')\defeq\left\{(\lambda,z):\lambda \in [0,1],z\geq 0,(p,q)\succ (\lambda p',z q')\right\}.
\end{align}
Its boundary can be specified by giving the largest possible $\lambda$ for each $z$; call this value $\lambda^\star_z(R\to R')$. 
We will also have use to refer to the smallest $z$ for a given $\lambda$, call it $z^\star_\lambda(R\to R')$. 
The boundary can be determined directly from the Lorenz curves of $R$ and $R'$, as illustrated in Fig.~\ref{fig:aj}.
\begin{corollary}
\label{cor:boundary}
For any normalized $R$ and $R'$, the set $\mathcal S(R,R')$ is convex.
Furthermore, for $T^\star(p',q')=\{(x_k^\star,y_k^\star)\}_{k=1}^{n'}$,  
\begin{align}
\lambda^\star_z(R\to R')&= \min_{k\in 1,\dots,n'} \frac{\alpha_{y_k^\star}(p,\tfrac 1z q)}{\alpha_{y_k^\star}(p',q')}\label{eq:lambdaofzmath}.
\end{align}
\end{corollary}
\begin{proof}
To establish convexity, suppose $M_j$ is a substochastic matrix associated with $(p,q)\succ (\lambda_j p',z_j q')$ according to \eqref{eq:Lorenzmajorizationdef}. 
Any convex combination of $M_j$'s is also substochastic, and immediately gives $(p,q)\succ (\bar \lambda p',\bar z q')$, where $\bar \lambda$ ($\bar z$) is the same convex combination of the $\lambda_j$ ($z_j$). 

To establish \eqref{eq:lambdaofzmath}, observe that $(p,q)\succ (\lambda p',zq')$ is equivalent to $(p,\tfrac 1z q)\succ (\lambda p',q')$. By Theorem~\ref{thm:rsm}, this is further equivalent to $\alpha_y(p,\tfrac 1z q)\geq \alpha_y(\lambda p',q')=\lambda \alpha_y(p',q')$ for all $y\in [0,1]$. 
Thus, any feasible $\lambda$ must be less than the smallest ratio of the two Lorenz curve values, and we need only check the ratio at the elbows of $(p',q')$.
\end{proof}

Convexity of $\mathcal S(R,R')$ means that the function $z\mapsto \lambda^\star_z(R\to R')$ is concave; it is an increasing function since increasing $z$ in $(p,q)\succ (\lambda p',zq')$ still produces a feasible $\lambda,z$ pair. 
In the special case $R'=R$, it is not difficult to show 
\begin{align}
\label{eq:nofreelunch}
\lambda_z^\star(R\to R)=\left\{\begin{array}{ll} z & z\leq 1,\\ 1 &z>1.\end{array}\right.
\end{align}  
In terms of the Lorenz curve, it holds that $\alpha_y(p,q)\geq \alpha_y(\lambda p,zq)$ for all $y\geq 0$. 
At the point $y=z q_{\pi(1)}$, for $\pi$ a permutation as in Lemma~\ref{lem:Tcornerpoints}, we have $\alpha_y(\lambda p,zq)=\lambda p_{\pi(1)}$. 
However, for $y\in [0,q_{\pi(1)}]$, the function $y\mapsto \alpha_y(p,q)$ interpolates linearly between zero and $p_{\pi(1)}$.
Thus, $\alpha_y(p,q)=zp_{\pi(1)}$, and the above inequality becomes $z\geq \lambda$.
Clearly $\lambda =z$ is feasible, by choosing $M=\lambda \mathbbm 1$ in \eqref{eq:Lorenzmajorizationdef}.  

By the definition of submajorization, $(p,q)\succ (p',q')$ is equivalent to $(\lambda p,zq)\succ (\lambda p',zq')$, and the following proposition is immediate.  
\begin{proposition}
\label{prop:chain}
Suppose $(p_1,q_1)\succ (\lambda p_2, z q_2)$ and $(p_2,q_2)\succ (\lambda' p_3, z' q_3)$.
Then $(p_1,q_1)\succ (\lambda \lambda' p_3,zz' q_3)$. 
\end{proposition}


\subsection{Approximate relative submajorization}
It is also interesting to consider the question of approximate relative submajorization. 
Even if $(p,q)\nsucc (p',q')$, it might be possible to find some other pair $(\hat p',\hat q')$ close to $(p',q')$ such that $(p,q)\succ (\hat p',\hat q')$ as quantified by some measure.
Let us specialize to the variational distance $\delta$ and measure the approximation quality $\delta(p',\hat p')$ and $\delta(q',\hat q')$ separately. (Note, we are now using the variational distance outside the usual context of normalized vectors, but this will be useful later.) Then we can make the following definition.
\begin{definition}[Approximate relative submajorization]
The pair $(p,q)\in \mathbb R^{2n}_+$ is said to approximately relatively submajorize $(p',q')\in \mathbb R_+^{2n'}$ with errors $\eps$ and $\eta$, denoted $(p,q)\succ_{\eps,\eta} (p',q')$, when there exist $\hat p'$ and $\hat q'$ such that $(p,q)\succ (\hat p',\hat q')$, $\delta(p',\hat p')\leq \eps$, and $\delta(q',\hat q')\leq \eta$. 
\end{definition}
More explicitly, using the dual formulation of the variational distance, we have $(p,q)\succ_{\eps,\eta} (p',q')$ when there exists a substochastic $n'\times n$ matrix $M$ and $a,b\in \mathbb R_+^{n'}$ such that
\begin{subequations}
\label{eq:approxLorenzdef}
\begin{align}
1^T a&\leq \eps,\\
1^T b&\leq \eta,\\
Mp&\geq p'-a,\\
Mq&\leq q'+b.
\end{align}
\end{subequations}

As in the previous section, the feasibility of any particular $\eps$ and $\eta$ combination for two given pairs $R=(p,q)$ and $R'=(p',q')$ can be determined directly from the Lorenz curves of $R$ and $R'$.
This is the content of the following theorem, whose proof(\hyperlink{proof:th2}{$\rightarrow$}) is given in the Appendix.
\begin{theorem}
\label{thm:approxLorenzbeta}
For any $p,q\in \mathbb R_+^n$ and $p',q'\in \mathbb R_+^{n'}$ and $\eps,\eta\geq 0$, the following are equivalent:
\begin{enumerate}[(a)]
\item $(p,q)\succ_{\eps,\eta} (p',q')$,
\item $\beta_x(p,q)\leq \beta_{x+\eps}(p',q')+\eta$ for all $x\in \mathbb R$.
\end{enumerate}
\end{theorem}

In the context of statistics the definition above corresponds to the randomization characterization for approximate sufficiency of one experiment relative to another (see \cite[\S6.4]{torgersen_comparison_1991}), which is related to the decisition theoretic notion of deficiency of experiments by LeCam's randomization criterion~\cite[Th.\ 6.4.1]{torgersen_comparison_1991}. Theorem~\ref{thm:approxLorenzbeta} is a generalization of the analogous result for strict relative majorization by Torgersen~\cite[Cor.\ 16]{torgersen_comparison_1970},\cite[Th.\ 5.2]{torgersen_stochastic_1991},\cite[Cor.\ 9.3.26]{torgersen_comparison_1991}.

\section{Applications to the resource theory of thermal operations}
\label{sec:resourcethermo}
\subsection{Resource theories and thermal operations}
\label{sec:thermalops}
A useful approach to understanding thermodynamics, and particularly the thermodynamics of systems we describe using microscopic degrees of freedom, is as a \emph{resource theory}~\cite{janzing_thermodynamic_2000,brandao_resource_2013}. 
Any resource theory posits a restriction on the possible operations available to an experimenter and then studies which states can be transformed into which others under this restriction. 
For instance, by restricting operations to those which have a given symmetry, systems that do no possess the symmetry become a useful resource which can be used to implement previously-unavailable transformations. 
A well-studied example in quantum information theory is the resource theory of bipartite entanglement, corresponding to a restriction on bipartite quantum operations to those which consist of operations on one subsystem or the other, augmented by classical communication between the two (local operations and classical communication, LOCC). 
The maximally-entangled state is the ultimate resource in this theory, as it completely lifts the LOCC restriction. 
By making use of quantum teleportation, one can perform any bipartite quantum operation using LOCC operations and the maximally-entangled state.  

The resource theory approach to thermodynamics seeks to model thermal operations at a fixed background inverse temperature $\beta$ (not at all related to $\beta_x(p,q)$).
In the theory, each system has a Hamiltonian describing its allowed energies and specifying its dynamics, interactions are presumed to be negligible (else we would define the systems differently), and the only allowed operations are those which preserve the energy, i.e.\ those which commute with the total Hamiltonian. 
The background temperature is fixed by the additional allowed operation of creating the Gibbs state at inverse temperature $\beta$ of any system. 
Put differently, these systems comprise the thermal bath which fixes the temperature.
Any state that is not precisely of this form is regarded as a resource, and the central questions of the resource theory revolve around determining which resources can be transformed into which others, at what rates, and so forth. 
Recently, there have been a number of important results in the area~\cite{horodecki_fundamental_2013,brandao_second_2015,yunger_halpern_beyond_2016,alhambra_what_2015,wehner_work_2015,lostaglio_stochastic_2015}.
Note that in this paper, we shall consider only states which are block-diagonal in the energy eigenbasis, and therefore we do not consider the very delicate question of the impact of quantum coherence on the resource theory. 

It is sufficient to describe any resource system $R$ by the pair $(r,g)$ of its probability distribution of energy levels $r$ and the corresponding distribution for the Gibbs state $g$. 
We say that $R\geq R'$ if there exists a thermal operation of the kind described above which can transform $R$ into $R'$. 
Crucially for our present purposes, this condition is equivalent to relative majorization of the vector pairs, as shown in \cite{janzing_thermodynamic_2000} and \cite{horodecki_quantumness_2013} (in the latter, this is termed ``thermomajorization'').
Formally, we have 
\begin{proposition}[Theorem 5~\cite{janzing_thermodynamic_2000}]
\label{prop:thermomajorize}
For any resources $R$ and $R'$, 
\begin{align}
R\geq R' \quad \Leftrightarrow \quad(r,g)\succeq (r',g').
\end{align} 
\end{proposition}

\subsection{Probabilistic and work-assisted transformations}
\label{sec:probworktrans}
Relative submajorization is useful in the context of thermodynamics, as it succinctly encapsulates the possibility of probabilistic work-assisted thermal transformations. 
In this section we will show that Corollary~\ref{cor:boundary} gives a simple expression, in terms of their Lorenz curves, for the optimal probability of performing a transformation $R\to R'$ when aided by a given amount of work.
First let us precisely define probabilistic and work-assisted transformations separately and see how submajorization is relevant to each.
We begin with the latter.\\

If the transformation $R\to R'$ is not possible by thermal operations, we may consider supplying some amount of work so that the transformation is possible. 
In the resource theory setting, work can be understood as energy stored in a ``battery''  or ``accumulator'' system. 
The battery is a resource system consisting of many, finely-spaced energy levels, and we confine ourselves to operations which always leave it in a well-defined energy level.
This prevents the battery from being used as an entropy sink. 
Denoting the battery in energy level $E$ by $B_E$, an amount of work $W$ can be extracted during the transformation of $R$ to $R'$ when, for any $E$, 
\begin{align}
R\otimes B_E\geq R'\otimes B_{E+W}.
\label{eq:workassistdef}
\end{align}
Our convention is that positive $W$ is work extracted, and negative $W$ is work expended. 
It will be more convienient to use the quantity $z=e^{-\beta W}$, for which $z<1$ corresponds to work extracted and $z>1$ to work expended.

Using the results of \S\ref{sec:mathtools}, we can easily show 
\begin{proposition}
\label{prop:workassist}
For any resources $R$ and $R'$, the following are equivalent
\begin{enumerate}[(a)]
\item $R\otimes B_E\geq R'\otimes B_{E+W}$ for some $E$,
\item $R\otimes B_0\geq R'\otimes B_{W}$,
\item $(r,g)\succ (r',zg')$.
\end{enumerate}
\end{proposition}
First, let $a(E)$ be the probability vector with all its weight concentrated on energy level $E$, and call the thermal state of the battery $g_B$.  
Then, \eqref{eq:workassistdef} is equivalent to  
\begin{align}
(r\otimes a(E),g\otimes g_B)\succeq (r'\otimes a(E+W),g'\otimes g_B).\label{eq:workmajorize}
\end{align}
Converting to the Lorenz curve, we have $\beta_x(r\otimes a(E),g\otimes g_B)\leq \beta_x(r'\otimes a(E+W),g'\otimes g_B)$ for all $x\in [0,1]$.
Using Lemma~\ref{lem:betalemma}, one may simplify this condition to 
\begin{align}
\frac{e^{-\beta E}}{Z_B}\beta_x(r,g)\leq \frac{e^{-\beta (E+W)}}{Z_B}\beta_x(r',g'),
\end{align}
again for all $x\in [0,1]$, where $Z_B$ is the partition function of the battery. 
The factor $e^{-\beta E}/{Z_B}$ is evidently common to both sides, leaving just $z\beta_x(r',g')$ on the right.
But this is simply $\beta_x(r',zg')$, so Theorem~\ref{thm:rsm} 
implies \eqref{eq:workassistdef} is equivalent to $(r,g)\succ (r',zg')$ as claimed.    
Since the value of $E$ is ultimately irrelevant, our definition of work is equivalent to the ``work bit'' formulation of~\cite{horodecki_fundamental_2013}.
The connection between work assistance and rescaling of the Lorenz curve, which is implicit in Prop.~\ref{prop:workassist}, was observed in \cite{horodecki_fundamental_2013,renes_work_2014,yunger_halpern_beyond_2016}.\\

Similarly, if the $R\to R'$ transformation is impossible, we may still attempt to produce $R'$ probabilistically.
There are two sensible notions of probabilistic transformation, mixture probability and heralded probability, and we can show that in a certain setting these are identical.  
The first approach, taken in~\cite{alhambra_what_2015}, treats probability as the weight $\lambda$ of $r'$ in a mixture of $r'$ and some other state $s'$ at the output. 
That is, it is possible to transform $R$ into $R'$ with probability $\lambda$ when there exists some other resource $\hat R'=(s',g')$ such that 
\begin{align}
\label{eq:probtransdef}
R\geq \lambda R'+(1-\lambda)\hat R'.
\end{align}
In terms of majorization, this condition is $(r,g)\succeq (\lambda r'+(1-\lambda)s',g')$,
which implies $(r,g)\succcurlyeq (\lambda r',g')$ and therefore $(r,g)\succ (\lambda r',g')$.
The converse follows by the second part of Lemma~\ref{lem:directroute} in the Appendix.

Heralded probability refers to the setup in which the output is accompanied by an additional system which indicates whether $R'$ was successfully created or not. 
The additional system can simply be a two-level system with trivial Hamiltonian. 
Despite the fact that energy is not directly involved, such systems have resource value in the thermal setting since they can be used as entropy sinks, and so it is necessary to take this value into account. 
One way to do so is to simply provide a suitably-initialized bit at the beginning of the process. 
Initially the bit is in one of the two states, say \texttt 0, described by the probability distribution $b(\texttt 0)=(1,0)$, while the equilibrium state $g_b$ of the bit is simply the equal mixture of \texttt 0 and \texttt 1, $g_b=(1,1)/2$.
Defining $b_k=(b(k),g_b)$, it is possible to transform $R$ to $R'$ with heralded probability $p$ when 
\begin{align}
R\otimes b_0\geq p R'\otimes b_1+(1-p)\hat R'\otimes b_0,
\end{align}
again for some $\hat R'=(s',g')$. In terms of majorization, this condition is just
\begin{align}
(r\otimes b(0),g\otimes g_b)\succeq (pr'\otimes b(1)+(1-p)s'\otimes b(0),g'\otimes g_b).
\end{align}
But this is nothing other than (b) in Theorem~\ref{thm:rsm} (upon using $z=1$ and rescaling the second vector in each pair by $\frac 12$), and so the mixture and heralded notions of probability are in fact identical (when supplying the additional bit for the latter). 
Formally, we have shown
\begin{proposition}
\label{prop:probtrans}
For any two resources $R$ and $R'$, the following are equivalent:
\begin{enumerate}[(a)]
\item $R\geq \lambda R'+(1-\lambda)\hat R'$,
\item $R\otimes b_0\geq \lambda R'\otimes b_1+(1-\lambda)\hat R'\otimes b_0$,
\item $(r,g)\succ (\lambda r',g')$.
\end{enumerate}
\end{proposition}
That there exists a measurement transforming the mixture $\lambda r'+s'$ into the heralded mixture $\lambda r'\oplus s'$ was observed in~\cite{alhambra_what_2015}. 
Moreover, it was shown that there exists a thermal operation which can perform the measurement, provided a blank bit is made available at the outset of the process, but the operation involves a quantum unitary operation. 
It was left open whether or not the measurement can be performed by classical thermal operations, as we are considering here.  
Prop.~\ref{prop:probtrans} does not address this question directly, but does show that the heralded probability equals the mixed probability under classical thermal operations when a blank bit is supplied in the former case. 
It appears from numerical investigation that indeed it is possible to transform any mixture to a heralded mixture by classical thermal operations, again aided by a blank bit, but a proof of this statement (or a counterexample) must be left to future work.\\ 



It is a simple matter to now combine work-assistance with probabilitistic transformation. 
Using Prop.~\ref{prop:workassist}(b) and Prop.~\ref{prop:probtrans}(a), the transformation of $R$ to $R'$ with probability $\lambda$ and $z=e^{-\beta W}$ is possible when 
\begin{align}
\label{eq:defprobworktrans}
R\otimes B_0\geq \lambda R'\otimes B_W+(1-\lambda)R'',
\end{align}
where $R''$ is a joint resource state of the output and battery. 
That is, in this definition, the probability does not distinguish between failure to produce $R'$ and failure to leave the battery in the appropriate energy level.
In terms of submajorization, \eqref{eq:defprobworktrans} is 
\begin{align}
\label{eq:majprobworktrans}
(r\otimes a(0),g\otimes g_B)\succeq (\lambda r'\otimes a(W)+(1-\lambda)s',g'\otimes g_B)
\end{align}
for some $s'$. 
This implies 
\begin{align}
\label{eq:convergecondition}
(r\otimes a(0),g\otimes g_B)\succcurlyeq(\lambda r'\otimes a(W),g'\otimes g_B),
\end{align}
and then by converting to the Lorenz curve as in the case $\lambda=1$, we arrive at the conclusion that \eqref{eq:defprobworktrans} implies $(r,g)\succ(\lambda r',zg')$. 
Conversely, when the submajorization condition holds, then by Theorem~\ref{thm:rsm} there exists an $s'$ such that \eqref{eq:defprobworktrans} holds.
Indeed, from the statement of the theorem, it is always possible to pick an $s'$ with no support on the $a(W)$ sector. 


Now we may apply Corollary~\ref{cor:boundary} to obtain the optimal $\lambda$ for a given $z$ and vice versa. For  $T^\star(r',g')=\{(x_k^\star,y_k^\star)\}$, we have
\begin{align}
\lambda^\star_z(R\to R')&= \min_{k\in 1,\dots,n'} \frac{\alpha_{y_k^\star}(r,\tfrac 1z g)}{\alpha_{y_k^\star}(r',g')},\label{eq:lambdaofz}\\
z^\star_\lambda(R\to  R') &= \max_{k\in 1,\dots n'}\frac{\beta_{x_k^\star}(\tfrac 1\lambda r,g)}{\beta_{x_k^\star}(r',g')}.\label{eq:zoflambda}
\end{align}

It is interesting to note that a different form of the latter expression was obtained in a different model of thermal operations and a different notion of work by Egloff \emph{et al.}~\cite{egloff_measure_2015}. They use a more cumbersome integral form of the Lorenz curve, for which it is not clear that one can easily compute $z_\lambda^\star(R\to R')$. Lemma~\ref{lem:Tcornerpoints} shows that the Lorenz curve has a much simpler form, so that in \eqref{eq:lambdaofz} we need only minimize the ratio of the Lorenz curves over a finite number of points, where  the most computationally-intensive step of constructing each Lorenz curve is sorting the set of ratios $r_k/g_k$.
We shall discuss the relation between the two approaches in more detail in Appendix~\ref{sec:egloff}. 

In the setting of thermal operations using the same definition of work as here, special cases of this result have been previously derived. 
The case $\lambda=1$ in $z^\star_\lambda$ was observed to be expressible as a linear program by the author in~\cite{renes_work_2014},\footnote{Note that the model of work in~\cite{renes_work_2014} is to change the gap of a two level system, while ensuring that it remains in its excited state.
This model is only equivalent to the present formulation when the initial and final gaps are very large compared to the temperature. 
Additionally, one statement of \cite{renes_work_2014} is in error: It is possible to formulate the work cost in the basic setup by minimizing the variable $x=e^{-\beta(E+W)}/(1+e^{-\beta (E+W)})$. 
The constraint $G(1,e^{-\beta E})/(1+e^{-\beta E})\otimes g=(x,1-x)\otimes g'$ is now linear in $x$.} 
and was subsequently shown by Alhambra~\emph{et al.}~to have the form of \eqref{eq:zoflambda} in~\cite[Lemma 6]{alhambra_what_2015}. 
A particular example is depicted in Fig.~\ref{fig:workcost}. 
Meanwhile, the case $z=1$ in $\lambda_z^\star$ is Theorem 5 of~\cite{alhambra_what_2015}, an example of which is shown in Fig.~\ref{fig:probtrans}. 

Using \eqref{eq:lambdaofz} and \eqref{eq:zoflambda} we can quite easily rederive a variety of results on asymptotic transformations by simply invoking Stein's lemma. 
The details are presented in Appendix~\ref{sec:asymptotics}.

\begin{figure}[h]
\captionsetup[subfigure]{position=b}
\centering
\subcaptionbox{\label{fig:workcost}Work cost of $R\to R'$.
   }{\includegraphics{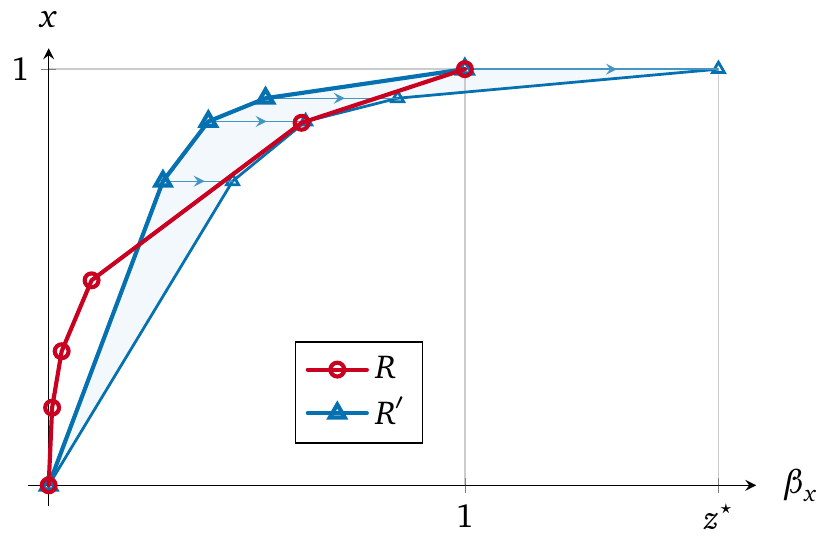}}
\hspace{3mm}
\subcaptionbox{\label{fig:probtrans}Transformation probability of $R\to R'$.
   }{\includegraphics{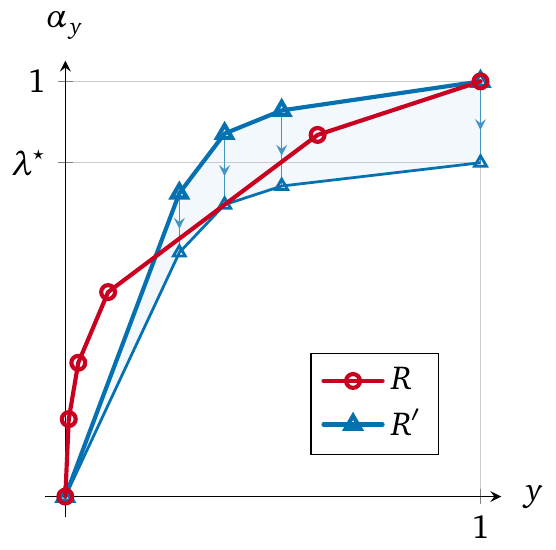}\par\vspace{1.2mm}
   }
  \caption{Determining (a) the work cost and (b) probability of the transformation $R\to R'$ via their Lorenz curves. 
   It is possible to transform $R$ to $R'$ while extracting work $W$ if and only if $(r,g)\succ (r',zg')$, where $z=e^{-\beta W}$ (cf.\ Prop.~\ref{prop:workassist}). 
   Thus, the optimal (smallest) $z$ is the largest ratio of the $\beta_x$ values of the input $R$ to the output $R'$ (cf.\ \eqref{eq:zoflambda} with $\lambda=1$).
   Similarly, it is possible to transform $R$ to $R'$ with probability $\lambda$ if and only if $(r,g)\succ (\lambda r',g')$, and the same holds for heralded transformations, provided a blank bit is provided at the outset of the process (cf.~Prop.~\ref{prop:probtrans}). Thus, the optimal (largest) $\lambda$ is the smallest ratio of the $\alpha_y$ values of the input to the output (cf.\ \eqref{eq:lambdaofz} with $z=1$). 
   }
\label{fig:workprob}
\end{figure}

\subsection{Approximate transformation}
\label{sec:approx}

We should be careful to contrast the above notions of \emph{probabilistic} transformation with that of an \emph{approximate} transformation of $R$ to $R'$. 
Given a distance measure $\delta$ between two probability distributions, there are indeed two distinct notions of an approximate transformation of $R$ to $R'$, corresponding to each of the two vectors in the resource pair. One could of course consider transformations which incur both kinds of error, but we will not pursue this further here.
The first kind of error, which we will denote $\eps$, pertains to transformations which take $g$ to $g'$ exactly, but produce only an approximate version of $r'$ from $r$.
The second kind of error, denoted $\eta$, is the opposite. 
The interpretation of the first is that the resource state of the output is approximated, while in the second the thermal state, or equivalently the Hamiltonian of the output system is approximated.  

It is not substantially more effort to formally consider work-assisted approximate transformations.\footnote{We will not consider approximate probabilistic work-assisted transformations for simplicity.}
We specialize to $\delta$ being the variational distance in the following. 
A particular $\eps$ is a feasible first kind of error when 
\begin{align}
R\otimes B_{E}\geq  \tilde R'_1\approx R'\otimes B_{E+W},
\end{align}
for some $\tilde R'_1=(s',g'\otimes g_B)$ with $\delta(r'\otimes a(E+W),s')\leq \eps$ and all $E$.
That is, the transformation produces an approximation to both the output system and the battery, and the approximation error does not distinguish between errors on either part. 
Similarly, a particular $\eta$ is a feasible second kind of error when 
\begin{align}
R\otimes B_{E}\geq \tilde R'_2\approx R'\otimes B_{E+W},
\end{align}
for some $\tilde R'_2=(r'\otimes a(E+W),s')$ with $\delta(g'\otimes g_B,s')\leq \eta$. 

Of course, we are mainly interested in the optimal errors for a given work-assisted transformation. 
We can use relative submajorization to characterize these without directly including the battery system as done above. 
It turns out that the structure of the battery and initial energy level is essentially irrelevant to errors of the first kind, but errors of the second kind depend on the initial energy level and the partition function $Z_B$ of the battery.
Letting $\eps_z^\star(R\to R')$ and $\eta_z^\star(R\to R')$ denote these optimal errors, we can show that they are directly related to the two parameters in approximate submajorization.
\begin{proposition}
\label{prop:approxRT}
For any resources $R$ and $R'$ and any $z\geq 0$, 
\begin{align}
\eps_z(R\to R')&=\min\{\eps:(r,g)\succ_{\eps,0} (r',zg')\}\quad\text{and}\label{eq:optimaleps}
\end{align}\\[-12mm]
\begin{subequations}
\label{eq:optimaleta}
\begin{align}
\eta_z(R\to R')&=\frac{e^{-\beta E}}{Z_B}\hat \eta_z(R\to R'),\quad \text{for}\\
\hat\eta_z(R\to R')&:=\min\{\eta: (r,g)\succ_{0,\eta} (r',zg')\}.
\end{align} 
\end{subequations}
\end{proposition}
\begin{proof}
Starting with \eqref{eq:optimaleps}, let $\hat s$ be such that $(r\otimes a(E),g\otimes g_B)\succeq (\hat s,g\otimes g_B)$ with $\delta (r'\otimes a(E+W),\hat s)=\eps$. 
We can express the variational distance in terms of the component of $\hat s$ with support on $a(E+W)$ in the battery, call it $s'$: 
\begin{align}
\delta(r'\otimes a(E+W),\hat s)
&=1\cdot (r'\otimes a(E+W)-\hat s)_+\\
&=1\cdot (r'-s')_+\\
&=\delta(r',s').
\end{align}
Moreover, by Theorem~\ref{thm:rsm}, $(r\otimes a(E),g\otimes g_B)\succeq (s'\otimes a(E+W)\oplus \bar s',g\otimes g_B)$ is equivalent to $(r,g)\succ(s',zg')$, where $\bar s'\oplus s'=\hat s$. 
This establishes \eqref{eq:optimaleps}.

The proof of \eqref{eq:optimaleta} is more involved, and is based on the proof of Lemma~\ref{lem:uglydual}.
As this lemma appears in the Appendix, we defer the present proof(\hyperlink{proof:prop5}{$\rightarrow$}) to the Appendix as well. 
\end{proof}

It is easy to show that $\eps^\star_z(R\to R')\leq 1-\lambda^\star_z(R\to R')$ by directly calculating the variational distance $\delta(r',\lambda r'+s')$; note that there is no difference between the two quantities when the output space has dimension two. 

Theorem~\ref{thm:approxLorenzbeta} immediately implies that the two optimal errors are equal to the maximal horizontal and vertical differences between the Lorenz curves, as follows: 
\begin{proposition}
\label{prop:approxbounds}
For any resources $R$ and $R'$,
\begin{subequations}
\begin{align}
\label{eq:epsdiff}
\eps^\star_z (R\to R')&= \max_{y\in \mathbb R}\left(\alpha_y(r',zg')-\alpha_y(r, g)\right),\\
\label{eq:etadiff}
\hat\eta^\star_z (R\to R')&= \max_{x\in \mathbb R} \left(\beta_x(r, g)-\beta_x(r',zg')\right).
\end{align}
\end{subequations}
\end{proposition}
Naturally, we need only optimize over the elbows of the two functions in either case, which follows by considering the derivatives of their difference.  

\subsection{Work value and cost}
Let us apply the above results to the question of the work value or work cost of a given resource. 
Here all the quantities of interest turn out to be simple functions of the Lorenz curve or its Legendre transform. 
The work value of $R$ corresponds to the work that can be gained in the transformation $R\to 1$, where $1$ denotes the trivial resource with $r'=g'=1$, i.e.\ no output. 
Conversely, the work cost of $R$ is just the work required to effect $1\to R$. 


Starting with work value, we find 
\begin{proposition}
\label{prop:workvalue}
For any resource $R$ and $z\in (0,1]$,
\begin{enumerate}[(a)]
\item $z^\star_\lambda(R\to 1)=\beta_\lambda(r,g),$
\item $\lambda^\star_z(R\to 1)=\alpha_z(r,g)=1-\eps^\star_z(R\to 1),$
\item $\hat\eta^\star_z(R\to 1)=1-z.$
\end{enumerate}
\end{proposition}
Note that (a) is the first part of Theorem 5 in~\cite{yunger_halpern_beyond_2016}. 
\begin{proof}
To start, set $r'=g'=1$ in \eqref{eq:zoflambda}, which means the only available $x_k^\star$ is $x_k^\star=1$. 
The statement (a) then follows, since $\beta_1(\tfrac 1\lambda r,g)=\beta_\lambda(r,g)$. 
For the first equality in (b), note that \eqref{eq:lambdaofz} is just $\alpha_1(r,\tfrac 1z g)=\alpha_z(r,g)$. 
The second must hold because only two levels of the work system are relevant, but it can also be seen directly by noticing that $y=z$ is the optimal choice in \eqref{eq:epsdiff}. 
Finally, (c) follows from \eqref{eq:etadiff} upon observing that $x=1$ is the optimal choice.
\end{proof}


To state the results for work cost, let $\phi_z(R)=\sum_{k}(r_k-zg_k)_+$; in the proof of Lemma~\ref{lem:directroute} it is established that $m\mapsto m\phi_{1/m}$ is the Legendre transform of $x\mapsto \beta_x$. 
Then we have
\begin{proposition}
\label{prop:workcost}
For any resource $R$,
\begin{enumerate}[(a)]
\item $z^\star_\lambda(1\to R)=\lambda \max_k \frac{r_k}{g_k}$,
\item For $z\geq 1$, $\eps^\star_z(1\to R)=\phi_{z}(R)=\eta^\star_z(1\to R)$. 
\end{enumerate}
\end{proposition}
\begin{proof}
For (a), interchange $R$ and $R'$ in \eqref{eq:zoflambda}, and set $(r',g')=(1,1)$. 
Since $\beta_x(\frac1\lambda 1,1)=\lambda x$, we find
\begin{align}
z^\star_\lambda (1\to R)
&=\max_{k\in 1,\dots n}\lambda \frac{\alpha_{y_k^\star}(r,g)}{\beta_{x_k^\star}(r,g)}\\
&=\lambda\frac{\alpha_{y_1^\star}(r,g)}{\beta_{x_1^\star}(r,g)}\\
&=\lambda \max_k \frac {r_k}{g_k}.
\end{align}

For the second equality in (b), using \eqref{eq:etadiff} gives 
\begin{align}
\hat \eta_z^\star(1\to R)&=\max_{x\in \mathbb R}(x-z\beta_x(r,g))\label{eq:46}\\
&=z\max_{x\in \mathbb R}(\tfrac xz -\beta_x(r,g))\\
&=\phi_z(R).
\end{align}
Similarly, but slightly more involved, starting from \eqref{eq:epsdiff} we have
\begin{align}
\eps_z^\star(1\to R) &=\max_{y\in\mathbb R}(\alpha_y(r,zg)-y)\\
&=\max_{y\in\mathbb R}\, \max_{t:t\cdot zg=y} (t\cdot r-z t\cdot g)\\
&=\max_{0\leq t\leq 1_n} (t\cdot r-z t\cdot g)\\
&=\max_{x\in \mathbb R}\, \max_{t:t\cdot r=x}(t\cdot r-z t\cdot g)\\
&=\max_{x\in \mathbb R}\, (x-z\beta_z(r,g)).
\end{align}
This is just \eqref{eq:46}, and the remaining steps go through as before. 
\end{proof}

Since the work value and work cost involve functions which are Legendre transforms of one another, we can apply Fenchel's inequality~\cite[\S3.3.2]{boyd_convex_2004} to obtain the following
\begin{proposition}
\label{prop:fenchel}
For any resource $R$ and $z,z'\geq 0$,
\begin{align}
\eps_z^\star(R\to 1)+\eps^\star_{z'}(1\to R)\geq 1-zz',
\end{align}
with equality if $z'$ is the slope of a tangent to $\alpha_z(r,g)$ at $z$. 
\end{proposition}
The case of equality is illustrated in Figure~\ref{fig:Fenchel}.
\begin{proof}
Start with the direct statement of Fenchel's inequality:
\begin{align}
\beta_x(r,g)+m\phi_{\frac 1m}(R)\geq xm.
\end{align}
Substituting $z'=\frac 1m$ and working in terms of $\alpha_z(r,g)$ instead yields
\begin{align}
z+\frac 1{z'}\phi_{z'}(R)\geq \alpha_z(r,g)\frac 1{z'}.
\end{align}
Rearranging and using Propositions~\ref{prop:workvalue} and~\ref{prop:workcost} gives the result. 
\end{proof}

\begin{figure}
{\centering
\includegraphics{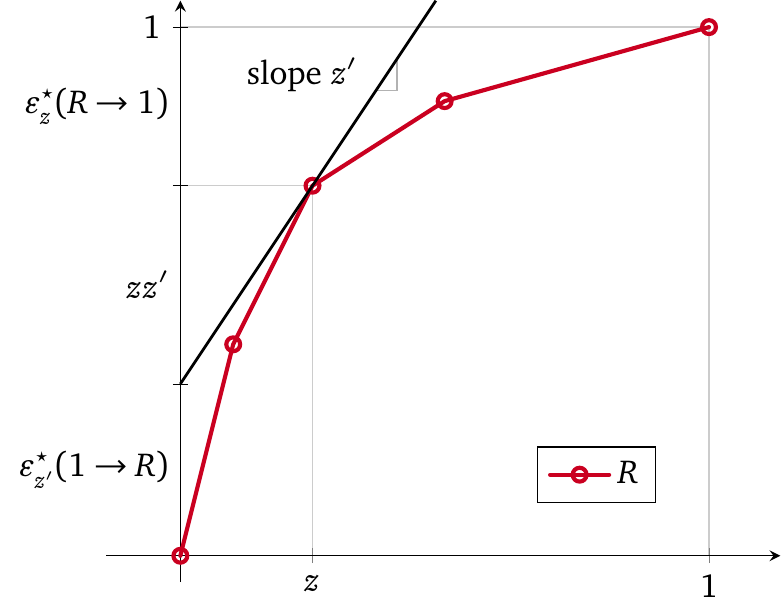}

}
\caption{\label{fig:Fenchel}Graphical depiction of the equality condition in Prop.~\ref{prop:fenchel}. Choosing $z$ and $z'$ such that a line with slope $z'$ is tangent to $\alpha_z(r,g)$ leads to the equality $\eps_z^\star(R\to 1)+\eps_{z'}^\star(1\to R)+zz'=1$ for any $R=(r,g)$. } 
\end{figure}

\subsection{Useful bounds}

In this section, we detail some interesting and useful bounds that one can obtain from relative submajorization. 
First we discuss constraints on work, transformation probability, and approximation error in transformations taking $R$ to $R'$. 
Later, we consider bounds involving these quantities for several different transformations, including $R\to R'$ and its reverse. 

The first two bounds provide a way to bootstrap from a pair $(\lambda,z)$ known to be feasible or infeasible in the $R\to R'$ transformation to some constraints on the optimal values at given $\lambda$ or $z$. 
\begin{proposition}
For any resources $R$ and $R'$, suppose $(\lambda,z)$ is feasible in $(r,g)\succ (\lambda r',zg')$. Then, 
\begin{align}
\label{eq:feasbound}
\lambda_z^\star(R\to R')\,z^\star_\lambda(R\to R')\leq \lambda\, z.
\end{align}
On the other hand, if $(\lambda,z)$ is infeasible, then
\begin{align}
\label{eq:infeasbound}
\lambda_z^\star(R\to R')\,z^\star_\lambda(R\to R')\geq \lambda\, z.
\end{align}
\end{proposition}
As an illustration, take the case $\lambda=z=1$. 
If $R\ngeq  R'$, then the second of these implies $\lambda^\star\, z^\star\geq 1$ (this is the lower bound of Lemma 7 in~\cite{alhambra_what_2015}). 
The optimal transformation probability provides a lower bound on the optimal $z$, hence a lower bound on the optimal amount of work which must be expended to drive the transformation, and vice versa.
On the other hand, if $R\geq R'$, then any $(\lambda,1)$ is feasible, and \eqref{eq:feasbound} implies $z^\star_\lambda\leq \lambda$.
In words, a decreased success probability leads to a gain in work extractable from the transformation which is at least in proportion (strictly, $z=e^{-\beta W}$ is in proportion) to the probability. 
\begin{proof}
These properties follow directly from the fact that $z\mapsto \lambda_z^\star$ is an increasing, concave function. 
For the first, given a feasible $(\lambda, z)$, consider the line joining the origin with $(\lambda_z^\star,z)$; the boundary of the feasible region must lie above this line.
Interpolating linearly from $(\lambda_z^\star,z)$ to the point $(\lambda,z')$ gives a feasible $z'=\lambda z/\lambda_z^\star$. 
Then \eqref{eq:feasbound} is a rearrangement of $z^\star_\lambda\leq z'$. 

Now suppose instead that $(\lambda,z)$ is infeasible.
The line joining the origin with $(\lambda^\star,z)$ intersects the boundary at $(\lambda^\star,z)$ and lies above it for larger $z$.
Interpolating linearly to $(\lambda,z')$ again results in $z'=\lambda z/\lambda_z^\star$, but now $z^\star_\lambda\geq z'$, which is \eqref{eq:infeasbound}.   
\end{proof}


Note that \eqref{eq:infeasbound} implies $(1-\eps^\star_z(R\to R'))\,z^\star_\lambda(R\to R') \geq \lambda z$ for $(\lambda,z)$ infeasible. 
The following is a similar version for $\hat \eta_z^\star$, with the restriction that $(1,z)$ be infeasible. 
It provides an upper bound on the optimal $\eta$ error in terms of the optimal work extracted in the transformation.  
\begin{proposition}
\label{prop:etazbound}
For any resources $R$ and $R'$ and $z\leq z^\star(R\to R')$, 
\begin{align}
(1-\hat\eta_z^\star(R\to R'))\,z^\star(R\to R') &\geq z.\label{eq:etazgeq1}
\end{align}
\end{proposition}
We defer the proof(\hyperlink{proof:prop11}{$\rightarrow$}) to the Appendix, as it makes use of the linear program formulation of $\hat\eta_z^\star$. 


Next, we consider relations between work, success probability, and approximation error for several transformations, in particular $R\to R'$ and the reverse transformation $R'\to R$.
First, Prop.~\ref{prop:chain} gives a kind of ``chain rule'' for work cost and transformation probability. 
\begin{proposition}
For any three resources $R_1$, $R_2$, and $R_3$, $\lambda,\lambda'\in [0,1]$ and $z,z'\in \mathbb R_+$,
\begin{align}
\label{eq:workchain}
z^\star_{\lambda\lambda'}(R_1\to R_3)&\leq z^\star_\lambda(R_1\to R_2)z^\star_{\lambda'}(R_2\to R_3),\\
\lambda^\star_{zz'}(R_1\to R_3)&\geq \lambda^\star_z(R_1\to R_2)\lambda^\star_{z'}(R_2\to R_3).\label{eq:lambdachain}
\end{align}
\end{proposition}
\begin{proof}
In the former, for fixed $\lambda,\lambda'$, Prop.~\ref{prop:chain} applied to $z=z^\star_\lambda(R_1\to R_2)$ and similarly for $z'$ implies   
$z^\star_\lambda(R_1\to R_2)z^\star_{\lambda'}(R_2\to R_3)$ is feasible for transforming $R_1$ to $R_3$ with probability $\lambda \lambda'$. 
But the optimal value of $z$ for this transformation must be smaller, whence~\eqref{eq:workchain}. 
Entirely similar reasoning gives \eqref{eq:lambdachain}.
\end{proof}

Setting $R_3=R_1$ in Prop.~\ref{prop:chain} and using \eqref{eq:nofreelunch}, one generally has $\lambda \lambda'\leq zz'$ for any $(\lambda,z)$ feasible in the forward transformation and $(\lambda',z')$ feasible in the reverse. 
Choosing various optimal parameters then gives the following:
\begin{proposition}
For any resources $R$ and $R'$, $\lambda,\lambda'\in [0,1]$, and $z\geq 0$,
\label{prop:reversibility}
\begin{align}
\label{eq:workchainAA}
z^\star_\lambda(R\to R')\,z^\star_{\lambda'}(R'\to R)&\geq \lambda\lambda',\\
\label{eq:lambdazforwardreverse}
\lambda \lambda^\star_z({R'\to R}) &\leq  z z^\star_\lambda({R\to R'}).
\end{align}
\end{proposition}

The first of these states that the work cost of the reverse process is bounded from below by a simple function of the work cost of the forward process and the probabilities fixed for each.
This is an improvement of the ``triangle inequality'' described in \cite{egloff_measure_2015}, as their bound effectively involved $z^\star_{\tilde\lambda}(R\to R)$ on the righthand side, with $\tilde \lambda=\lambda+\lambda'-1$.%
\footnote{Eq.\ \ref{eq:workchainAA} is, in their notation, $W^\eps(A\to B)+W^\eps(B\to A)\leq \tfrac 2\beta\log\frac{1}{1-\eps}$; on the right we now have a constant instead of the quantity $W^{2\eps}(A\to A)$. 
Eq.~\ref{eq:workchain} is a slight improvement over \cite[Lem.\ 14]{egloff_measure_2015}: $\lambda\lambda'$ appears on the lefthand side as opposed to $\lambda+\lambda'-1$.
}  
The second of these gives an upper bound to the optimal probability of the reverse transformation in terms of the optimal work cost of the forward transformation.
The case $\lambda=z=1$ is the upper bound of Lemma 7 in~\cite{alhambra_what_2015}. 
Eq.~\eqref{eq:lambdazforwardreverse} is a very useful generalization, since the $\lambda=z=1$ bound is only meaningful when $z^\star(R\to R')< 1$, i.e.\ when work can be extracted in the $R\to R'$ transformation, which is not the generic situation. 
Here, allowing arbitrary $\lambda$ removes this restriction and can, by picking $\lambda$ small enough, yield a nontrivial bound on $\lambda^\star_z$ even when $R\ngeq R'$.

We can also show that a large probabilistic work gain in the $R\to R'$ transformation leads to high $\eta$ approximation error of the reverse transformation. 
\begin{proposition}
\label{prop:etareversibility}
For any resources $R=(r,g)$ and $R'=(r',g')$,
\begin{align}
 \hat \eta_z^\star(R'\to R)&\geq \frac 1z-z^\star({R\to R'}). \label{eq:zetareverse}
\end{align}
\end{proposition}
\begin{proof}
Using \eqref{eq:etadiff} we have
\begin{align}
\hat\eta^\star_z (R\to R')&= \max_{x\in [0,1]} \beta_x(r,g)\left(\frac 1z-\frac{\beta_x(r',g')}{\beta_x(r,g)}\right)\\
&\geq \max_{x\in [0,1]} \beta_x(r,g)\min_{x\in [0,1]}\left(\frac 1z-\frac{\beta_x(r',g')}{\beta_x(r,g)}\right)\\
&=\frac 1z-{z^\star(R'\to R)}.
\end{align}
The last step uses \eqref{eq:zoflambda}.
\end{proof}

Upper bounds on the approximation error of the reverse process can be obtained in terms of the relative entropy, by making use of results on the reversibility of stochastic operations~\cite{winter_stronger_2012}.
\begin{proposition}
\label{prop:epsrecoverability}
For any resources $R$ and $R'$ such that $R\otimes B_0\geq R'\otimes B_{W}$, 
\begin{align}
\eps^\star_{1/z}(R'\to R)&\leq 2\left(D(r,g)-D(r',zg')\right)^2.\label{eq:errorr}
\end{align}
If $R\geq R'$, then 
\begin{align}
\eta^\star(R'\to R)&\leq 2\left(D(g,r)-D(g',r')\right)^2\label{eq:errorg}.
\end{align} 
\end{proposition}
\begin{proof}
The proof is essentially just that of Statement 3 of \cite{winter_stronger_2012}, which we include as Lemma~\ref{lem:petz}.
For the first case, let $T$ be the stochastic map from $R\otimes B_E$ to $R'\otimes B_{E+W}$. 
Then $\hat T$ from Lemma~\ref{lem:petz} takes $g'\otimes g_B$ to $g\otimes g_B$, so it corresponds to a thermal operation, and we have
\begin{align}
D(r\otimes a(E),g\otimes g_B)-D(r'\otimes a(E+W),g\otimes g_B)\geq D(r\otimes a(E),\hat T T r\otimes a(E)).
\end{align}
The left side can be simplified by using the facts $D(p\oplus 0,q\oplus q')=D(p,q)$ and $D(p,cq)=\log \frac 1c+D(p,q)$ for normalized $p$.
Then we obtain
\begin{align}
D(r,g)-D(r',zg)\geq D(r\otimes a(E),\hat T (r'\otimes a(E+W)).
\end{align}
Applying Pinsker's inequality, $\delta(p,q)\leq 2D(p,q)^2$~\cite[Lem.\ 11.6.1]{cover_elements_2006}, to the statement of the lemma gives \eqref{eq:errorr}. 

The second case is the same as the first, with the roles of $r$ and $g$ reversed. 
Since the relative entropy is undefined (infinite) when the support of the second argument is strictly contained in the support of the first, we must restrict to the case that $W=0$. 
\end{proof}

Wehner \emph{et al.}\ show \eqref{eq:errorr} for $z=1$ in the quantum case, even with catalytic operations, by directly constructing the thermal operation taking $R'$ to $R$~\cite{wehner_work_2015}.
Here, by constrast, we need not construct the thermal operation itself, just the apprpropriate stochastic map on $r'$ and $g'$. 
This approach gives a statement for $z\neq 1$ and naturally yields \eqref{eq:errorg}.

\section{Applications to the resource theory of pure bipartite entanglement}
\label{sec:resourcecoherence}
Another resource theory in which we may fruitfully apply relative submajorization is the resource theory of (pure) bipartite entanglement. 
This resource theory, introduced in~\cite{bennett_concentrating_1996,bennett_mixedstate_1996,bennett_purification_1996}, seeks to capture the difficulty of creating entangled states shared by two separated parties who can only perform local operations and exchance classical communication (LOCC operations). 

The general resource theory considers transformations of mixed states under LOCC.
When restricting to pure states, Nielsen showed that there exists an LOCC transformation taking one pure bipartite state to another exactly when the Schmidt coefficients of the latter majorize those of the former~\cite{nielsen_conditions_1999}. This statement is the analog of Prop.~\ref{prop:thermomajorize}. 
To formalize it, let $\ket{\psi_p}_{AB}$ be a bipartite state with Schmidt coefficients given by $p\in S_n$, i.e.\ $\ket{\psi_p}=\sum_{j=1}^n \sqrt{p_j}\ket{\xi_j}_A\otimes \ket{\eta_j}_B$ for some orthonormal bases $\{\ket{\xi_j}_A\}_j$ and $\{\ket{\eta_j}_B\}_j$.
Nielsen showed
\begin{proposition}[\cite{nielsen_conditions_1999}]
For any $p\in S_n$ and $q\in S_m$, there exists an LOCC operation $\mathcal E_{AB}$ mapping $\ket{\psi_q}_{AB}$ to $\ket{\psi_p}_{AB}$ if and only if $p\succeq q$.
\end{proposition}

Applying the results of relative submajorization, from this we immediately obtain the probability of a transformation between two given states, first shown in~\cite{vidal_entanglement_1999}.
If there is no LOCC operation taking $\ket{\psi_q}$ to $\ket{\psi_p}$, we may enlarge the output space by including an additional qubit $A'$ held by Alice and instead ask for an operation which produces $\sqrt{\lambda}\ket{\psi_p}\ket{0}+\sqrt{1-\lambda}\ket{\psi_r}\ket{1}$, where $r\in S_m$ is arbitrary. 
Measurement of the additional qubit in the classical basis produces $\ket{\psi_p}$ with probability $\lambda$.
The Schmidt coefficients of this superposition state are given by the classical probability distribution $\lambda p\oplus (1-\lambda)r$, and therefore the largest probability of the transformation is 
\begin{align}
\lambda^\star(\ket{\psi_q}\to \ket{\psi_p}):=\max\{\lambda:\lambda p\oplus (1-\lambda)r\succeq q, r\in S_{m}\}.
\end{align}
Note that this is not at all the analog of \eqref{eq:probtransdef}; the sense of majorization is opposite here. 
In terms of submajorization, the condition is just $(1_{n'},\lambda p)\succ (1_{n'},q)$, for $n'=\max(n,m)$. 
As in the proof of Corollary~\ref{cor:boundary}, we have $\lambda\beta_x(1_{n'}, p)=\beta_x(1_{n'},\lambda p)\leq \beta_x(1_{n'},q)$ for all $x\in \{1,\dots,n'\}$, and therefore 
\begin{align}
\lambda^\star(\ket{\psi_q}\to \ket{\psi_p})=\min_{x\in \{1,\dots,n'\} }\frac{\beta_x(1_{n'},q)}{\beta_x(1_{n'}, p)}.
\end{align} 
Since $\beta_k(1_{n'},q)=1-\sum_{j=1}^k q^\downarrow_j=\sum_{j={n'}-k}^{n'}q^\downarrow_j$, we have reproven
\begin{proposition}[Equation 3~\cite{vidal_entanglement_1999}]
For any $\ket{\psi_p}$ and $\ket{\psi_q}$, 
\begin{align}
\lambda^\star(\ket{\psi_q}\to \ket{\psi_p})=\min_{k\in \{1,\dots,n'\}}\frac{\sum_{j=k}^{n'}q^\downarrow_{j}}{\sum_{j=k}^{n'} p^\downarrow_{j}}.
\end{align}
\end{proposition}

Using the Lorenz curve, we may also determine the entanglement cost of a transformation, this time in an analogous fashion to the work cost of a thermal operation. 
To properly define the entanglement cost by analogy, we need to first establish a useful notion of an entanglement battery. 
Note that any state $\ket{\psi_p}$ for $p\in S_n$ can be created from the state $\ket{\psi_w}$ with $w=1_n/n$, because $p\succeq w$ holds for all $p\in S_n$.
Let us denote this state by $\ket{\phi_{n}}$. 
We may therefore consider the entanglement battery to be a system which is always in the state $\ket{\phi_{n_b}}$ for some $n_b$, and which may be used to enable an otherwise impossible transformation $\ket{\psi_q}\to \ket{\psi_p}$ in the sense that 
$\ket{\psi_q}\otimes \ket{\phi_{n_b}}\to \ket{\psi_{p}}\otimes \ket{\phi_{n'_b}}$
is possible for appropriate choices of $n_b$ and $n'_b$. 
In terms of majorization, the transformation is possible if and only if $p\otimes w_{n'_b}\succeq q\otimes w_{n_b}$.
As in the thermal case, we define 
\begin{align}
z^\star(\ket{\psi_q}\to \ket{\psi_p})\defeq \min\left\{\frac{n_b}{n'_b}:p\otimes w_{n'_b}\succeq q\otimes w_{n_b}\right\}.
\end{align}
Measuring entanglement logarithmically, the entanglement gain of the transformation is then $-\log z^\star(\ket{\psi_q}\to \ket{\psi_p})$. 
This notion has also been recently employed in \cite{egloff_measure_2015,alhambra_what_2015}. 

The majorization condition is equivalent to $\beta_x(p\otimes w_{n'_b},1_\infty)\leq \beta_x(q\otimes w_{n_b},1_\infty)$ for all $x\in [0,1]$. 
Here we embed both vectors in a countably infinite dimensional vector space rather than specifically naming the (smallest) dimension necessary to contain both input and output vectors. 
We can regard $1_\infty$ in this expression equally well as $1_\infty\otimes 1_{n_b}=n_b 1_\infty\otimes w_{n_b}$ and then use the fact that $\beta_x(p\otimes r,q\otimes r)=\beta_x(p,q)$ from Lemma~\ref{lem:betalemma} to infer that the above transformation is possible if and only if $n_{b}'\beta_x(p,1_\infty)\leq n_b \beta_x(q,1_\infty)$ for all $x\in[0,1]$. 
Equivalently, as in the thermal case, $z^\star(\ket{\psi_q}\to \ket{\psi_p})=\min\{z:(p,1_{\infty})\succ (q,z1_{\infty})$. 
We have shown the following analog to \eqref{eq:zoflambda} (with $\lambda=1$):
\begin{proposition}
\label{prop:coherencecost}
For any states $\ket{\psi_p}$ and $\ket{\psi_q}$, 
\begin{align}
z^\star(\ket{\psi_q}\to \ket{\psi_p})=\max_{x\in [0,1]}\frac{\beta_x(p,1_\infty)}{\beta_x(q,1_\infty)}.
\end{align}
\end{proposition}

Unlike in the resource theory of thermodynamics, here we cannot use Lorenz majorization to combine probabilistic and entanglement-assisted transformations. 
The difficulty is that both conditions lead to upper bounds, as in \eqref{eq:qcondition}, and neither to lower bounds. 

Continuing, we can formulate a useful notion of approximation by using the fidelity. 
Let $F^\star_z(\ket{\psi_q}\to \ket{\psi_p})$ be the largest fidelity between the output of an LOCC operation on $\ket{\psi_q}$ and the ideal output $\ket{\psi_p}$, when assisted by an amount of entanglement characterized by $z$.  
Note that that the fidelity of $\ket{\psi_p}$ and $\ket{\psi_{p'}}$ is given by 
\begin{align}
F(\ket{\psi_p},\ket{\psi_{p'}})=\sum_k \sqrt{p_k p'_k},
\end{align}
where the optimization over LOCC operations removes any differences in the choice of local bases.   
The righthand expression is the fidelity of the classical distributions, more commonly known as the Bhattacharyya coefficient $B(p,p')$. 
It is related to the variational distance by the bounds
\begin{align}
1-B(p,p')\leq \delta(p,p')\leq \sqrt{1-B(p,p')^2}.
\end{align}

Now suppose that $p$ does not majorize $q$, but does majorize $q'$, with $\delta(q,q')\leq \eps$.
Applying the operation $\mathcal E$ that produces $\ket{\psi_p}$ from $\ket{\psi_{q'}}$ to $\ket{\psi_q}$ instead produces $\rho_{p'}$. 
By the monotonicity of fidelity,
\begin{align}
F(\rho_{p'},\proj{\psi_p}) &= F(\mathcal E(\proj{\psi_{q}}),\mathcal E(\proj{\psi_{q'}}))\\
&\geq F(\proj{\psi_{q}},\proj{\psi_{q'}})\\
&\geq 1-\eps.
\end{align}
Therefore, 
\begin{align}
1-F^\star_z(\ket{\psi_q}\to \ket{\psi_p})&\leq \min\{\eps: (p,1_\infty)\succ_{\eps,0}(q,z1_\infty)\}\\
&=\max_{y\in \mathbb R}\alpha_y(q,z1_\infty)-\alpha_y(p,1_\infty).
\end{align}
Using the explicit form for $\alpha_y$, we have found the following bound on the fidelity of a pure state transformation.
\begin{proposition}
For any  $\ket{\psi_q}$ and $\ket{\psi_p}$,
\begin{align}
F^\star_z(\ket{\psi_q}\to \ket{\psi_p})\geq 1-\max_k\,\frac 1z\sum_{j=1}^k q^\downarrow_j-p^\downarrow_j.
\end{align}
\end{proposition}
We can apply \eqref{eq:errorr} and the bound $(1-\eps_z^\star)\,z^\star\geq z$ from \eqref{eq:infeasbound} to obtain further bounds in terms of entropies or entanglement cost:
\begin{proposition}
\label{prop:coherencebounds}
For any $\ket{\psi_q}$ and $\ket{\psi_p}$, if $p\succeq q$, then 
\begin{align}
F^\star(\ket{\psi_p}\to \ket{\psi_q})\geq 1-2\left(H(p)-H(q)\right)^2.
\end{align}
For $z\leq z^\star(\ket{\psi_q}\to \ket{\psi_p})$,
\begin{align}
F^\star_z(\ket{\psi_q}\to \ket{\psi_p})\geq \frac z{z^\star(\ket{\psi_q}\to \ket{\psi_p})}.
\end{align}
\end{proposition}

We can apply all of the above in the resource theory of coherence, introduced by \cite{aberg_quantifying_2006,baumgratz_quantifying_2014}. 
The resource theory of coherence seeks to capture the difficulty of creating coherent superpositions of states drawn from a fixed basis, call it $\{\ket{k}\}$.
Classical transformations from one basis state to another are presumed to be easy, but everything else difficult. Hence, the allowed operations are restricted to the so-called \emph{incoherent operations}. 
These include preparing any incoherent mixture of classical basis states (called an incoherent state), and any quantum operation all of whose Kraus operators map incoherent states to incoherent states. 

Winter and Yang observe that this implies each Kraus operator has the form $K=\sum_j c(j)\ket{f(j)}\bra{j}$ for some function $f$ and some complex coefficients $c(j)$~\cite{winter_operational_2015}. When $f$ is one-to-one in each Kraus operator, they call the operation \emph{strictly incoherent}.
For pure state transformations, they show that there exists a strictly incoherent operation transforming $\ket{\psi}$ to $\ket{\phi}$ if and only if the squared amplitudes (in the classical basis) of $\ket{\phi}$ majorize those of $\ket{\psi}$.\footnote{Du \emph{et al.}~\cite{du_conditions_2015} claim majorization is necessary and sufficient for the existence of more general incoherent operations. However, Winter and Yang note that there appears to be a gap in their reasoning and that necessity of majorization is still an open question.}

\section{Summary \& Open questions}
\label{sec:summary}
We have carefully investigated the properties of relative submajorization and shown that not only is submajorization characterized by the ordering of the Lorenz curves, but that the Lorenz curves also characterize the set of feasible $\lambda$ and $z$ in $(p,q)\succ (\lambda p',zq')$ as well as the approximation error in transforming $(p,q)$ to $(p',q')$. 
This has immediate application in the resource theories of thermodynamics and of quantum coherence.
In the thermodynamic setting, we find exact expressions for the work cost of a probabilistic transformation and for both kinds of approximation errors.
These lead to a characterization of the Lorenz curve itself in terms of the work value and cost of a given resource, as well as interesting bounds on the reversibility of a transformation. 
In the setting of LOCC manipulations of pure bipartite states, we find an exact expression for the entanglement cost of a transformation as well as bounds on the fidelity of entanglement-assisted transformation. 
These bounds also apply to the transformation of pure states under strictly incoherent operations. 

Recently, Alhambra et al.~\cite{alhambra_second_2016} extended the resource theory formulation of thermodynamics to incorporate arbitrary probabilistic work batteries and  showed how the Jarzynski~\cite{jarzynski_nonequilibrium_1997} and Crooks~\cite{crooks_entropy_1999} fluctuation relations can be easily derived in this extended framework. It should be noted that their framework also applies to general quantum resource states, not just the quasi-classical setting considered in this paper; see also \AA{}berg~\cite{aberg_fully_2016} for similar results in this direction.
It would be interesting to extend the methods presented here to their setup. In the following we sketch how this might be done.

Whereas we have treated the battery simply as another resource, subject to the constraint of being in a well-defined energy level, the battery in \cite{alhambra_second_2016} is a special system whose interactions with the thermal bath and resource states is restricted in such a way that it cannot be used as an entropy sink. 
Specifically, the energy-preserving unitary operation that couples the bath, resources, and battery must commute with shifts of the battery energy (a constraint also used in a slightly different setting in~\cite{skrzypczyk_work_2014}). 
This definition of the battery is via dynamics, not properties of its state, which enables one to consider arbitrary distributions of the work value in a given process.\footnote{The framework of Egloff et al.\ \cite{egloff_measure_2015} also allows for arbitrary distributions of work, but their definition of work is different from either of these. See Appendix~\ref{sec:egloff}.}

Now the connection to thermomajorization when considering work assistance is lost, i.e.\ Proposition~\ref{prop:workassist} does not hold in their framework. 
However, they show that a modified majorization condition does hold. 
Instead of a stochastic map on resource systems, one considers a stochastic map from the input resource system to two outputs: the output resource system and  the battery system. 
Calling the map $\hat M$, its components $\hat M(j,w|k)$ are probability distributions over $j$ and $w$ for each $k$. Here $j$ refers to the $j$th energy level of the output system, $k$ corresponding similarly to the input system, and $w$ to the work taken up by the battery (we assume the possible values form a finite set). 
They show that $\hat M$ can be implemented by thermal operations if and only if 
\begin{align}
\label{eq:Gibbsstoch}
\sum_{k,w} e^{\beta w} \hat M(j,w|k) e^{-\beta E_k}=e^{-\beta E_j}\qquad \forall j,
\end{align} 
where $\beta$ is the (inverse) temperature of the bath. 
We can express this in a notation similar to that of relative (sub)majorization by letting  $p$ be the vector with components $e^{\beta w}$ and $b$ ($b'$) the vector with components $e^{-\beta E_k}$ for all input (output) energy values $E_k$.
The condition \eqref{eq:Gibbsstoch} then reads  
\begin{align}
\label{eq:Gibbsstochmat}
(\id \otimes p^T) \hat M b=b'.
\end{align}
From \eqref{eq:workmajorize} it follows that any feasible work-assisted transformation as defined in \S\ref{sec:probworktrans} is an instance of an transformation according to \cite{alhambra_second_2016}. 
Similarly, the maximal amount of extractable work with fixed probability, given in \eqref{eq:zoflambda}, must apply in the extended framework since the optimal $\hat M$ can be converted into a feasible $M$ for the relative submajorization condition $(r,g)\succ (\lambda r',zg')$. 
Equation~\eqref{eq:Gibbsstochmat} also implies that deciding on the existence of a thermal operation that takes $r$ to $r'\otimes q$, i.e.\ $Mr=r\otimes q$, for some distributions $r$ and $r'$ of the resource and $q$ of the battery is again a linear program like that of relative majorization in~\eqref{eq:majorizationdef}.
It should then be possible to incorporate probabilistic or approximate transformations as was done for submajorization, as well as to investigate constraints on the work system, e.g.\ the necessary fluctuations for given transformations.

\vspace{5mm}
\noindent {\bf Acknowledgments:}
The author is grateful to Philipp Kammerlander for a careful reading of an early draft and 
Philippe Faist for helping clear up the relation of the results reported herein to those of~\cite{egloff_measure_2015}.
This work was supported by the Swiss National Science Foundation (through the National Centre of Competence in Research `Quantum Science and Technology') and by the European Research Council (grant No. 258932).

\printbibliography[heading=bibintoc,title=References]

\appendix
\section{Linear programming}
\label{sec:lp}
Since we are considering linear programs whose variables are stochastic matrices, it is useful to present a more general formulation of linear programming than that of, e.g.\ Boyd and Vandenberghe~\cite{boyd_convex_2004}.
Here we follow Barvinok~\cite{barvinok_course_2002}, making some specializations for simplicity. 
The setting of linear programming is two inner product spaces with a (non-negative) convex cone defined in each.\footnote{Even this is not necessary; the presentation of Barvinok considers the more general setting of vector spaces and their duals.} 
Let $V_1$ and $V_2$ be two vector spaces with inner products $\langle \cdot,\cdot\rangle_1$,  $\langle \cdot,\cdot\rangle_2$, and $K_1$ and $K_2$ cones in $V_1$ and $V_2$, respectively. 
In this setting, a particular linear program is specified by vectors $c\in V_1$, $b\in V_2$ and a linear map $A:V_1\to V_2$. 
The primal form is 
\begin{align}
  &\begin{array}{r@{\,\,}rl}
     \gamma(A,b,c) = & \underset{x}{\rm maximum} &  \langle c,x \rangle_{1}\\
     &\text{\rm subject to} & b-Ax\in K_2,\\
     && x\in K_1.
     \end{array}
\end{align}
The dual form of the program involves the adjoint map $A^*:V_2\to V_1$:
\begin{align}
  &\begin{array}{r@{\,\,}rl}
     \zeta(A,b,c) = & \underset{y}{\rm minimum} &  \langle b,y \rangle_{2}\\
     &\text{\rm subject to} & A^*y- c\in K_1,\\
     && y\in K_2.
     \end{array}
\end{align}

Weak duality is the inequality $\zeta(A,b,c)\geq \gamma(A,b,c)$, and strong duality the equality. The former always holds, while the latter is assured by a feasible $x$ in the primal or $y$ in the dual~\cite[\S5.2.4]{boyd_convex_2004}, \cite[Th.\ 8.2]{barvinok_course_2002}.
The following complementary slackness conditions are necessary and sufficient for optimal solutions  $x^\star$ and $y^\star$ when $\zeta(A,b,c)=\gamma(A,b,c)$: 
\begin{align}
\langle A^*y^\star-c,x^\star\rangle_1&=0,\\
\langle b- Ax^\star,y^\star\rangle_2&=0.
\end{align}

Now consider the case in which $V_1$ is the space of $n_2\times n_1$ real-valued matrices, equipped with the Hilbert-Schmidt inner product $\langle M,M'\rangle_1=\tr[M^TM']$, $K_1$ is the cone of matrices with non-negative entries, while $V_2=\mathbb R^{n_2}$ with the usual Euclidean inner product and $K_2=\mathbb R_+^{n_2}$, the positive orthant in $V_2$.  
The mapping $M\mapsto y=Ma$ for some $a\in \mathbb R^{n_1}$ is a linear map, a valid map $A$. 
Treating vectors as $n\times 1$ matrices, i.e.\ column vectors, the adjoint $A^*$ is the mapping $A^*:y\mapsto ya^T$. 
This can be easily verified from the definition of the adjoint, 
\begin{align}
\langle y,AM\rangle_2=\langle A^*y,M\rangle_1.
\end{align}
Writing $C$ for the variable in $V_1$ to emphasize that it is a matrix, the primal and dual forms become
\begin{align}
  &\begin{array}{r@{\,\,}rl}
     \gamma(a,b,C) = & \underset{M}{\rm maximum} &  \tr[C^TM]\\
     &\text{\rm subject to} & Ma\leq b,\\
     && M\geq 0,
     \end{array}
\end{align}
and
\begin{align}
  &\begin{array}{r@{\,\,}rl}
     \zeta(a,b,C) = & \underset{y}{\rm minimum} &  b\cdot y \\
     &\text{\rm subject to} & ya^T\geq  C,\\
     && y\geq 0.
     \end{array}
\end{align}

For use in numerical algorithms, it is useful to convert the above to standard ``vector'' form. 
Here, $M$ becomes the vector $v(M)$ formed by flattening out the matrix, joining subsequent rows, whence the inner product is just $\tr[M^TM']=v(M)\cdot v(M')$.
It can then be verified that right multiplication by $a\in \mathbb R^{n_1}$ is (left) multiplication by $\id_{n_2}\otimes a^T$, where $a$ is treated as a column vector, i.e.\ $Ma=(\id_{n_2}\otimes a^T)v(M)$. 
Similarly, left multiplication by $b\in \mathbb R^{n_2}$ is (left) multiplication by $b^T\otimes \id_{n_1}$, i.e.\ $b^TM=(b^T\otimes \id_{n_1})v(M)$. 
Readers familiar with the formalism of quantum information theory will recognize $v(M)$ as just $A\otimes \id \ket{\Omega}$ for $\ket{\Omega}=\sum_{k=1}^{n_1}\ket{k}\otimes \ket{k}$, right multiplication as $\id\otimes \bra{a}$ and left multiplication as $\bra{b}\otimes \id$. 

\section{Proofs deferred from the main text}
The following lemmas underpin the proof of Theorem~\ref{thm:rsm}. 
We begin with the equivalence of $\succ$ and $\curlyeqsucc$, which is important in the dilation lemma that follows. 
\begin{lemma}
\label{lem:upperrm}
For any $p,q\in \mathbb R_+^n$ and  $p',q'\in \mathbb R_+^{n'}$,
\begin{align}
(p,q)\succ (p',q')\quad \Leftrightarrow \quad (p,q)\curlyeqsucc (p',q').
\end{align}
\end{lemma}
\begin{proof}
The reverse implication is immediate. For for the forward implication, consider the  feasibility LP for Lorenz majorization and its dual
\begin{align}
\label{eq:feasLPprimal}
&\begin{array}{r@{\,\,}rl}
    c_\succ=& \underset{M}{\rm infimum} &  0\\
     &\text{\rm subject to} & Mp\geq p'\\
     && Mq\leq q',\\
     && 1^T_mM\leq1^T_n\\
     &&M\geq 0,
\end{array}\\
\label{eq:feasLPdual}
&\begin{array}{r@{\,\,}rl}
     c_\succ=& \underset{u,v,w}{\rm supremum} &  u^Tp'-v^Tq'-1^T_nw\\
     &\text{\rm subject to} & - pu^T+qv^T+w 1_{n'}^T\geq 0,\\
     &&u,v,w\geq 0.
\end{array}
\end{align}
First observe that dropping the $u\geq 0$ constraint in the dual does not change its optimal value. To see this, suppose $u\ngeq 0$ is feasible and let $u'=(u)_+$ be its positive part. Since $q$, $v$, $w$, and $1$ are all positive, setting the negative parts of $u$ to zero does not affect feasibility. Moreover, replacing $u$ with $u'$ can at best  increase the objective function. 

But the dual of the dual program absent the $u\geq 0$ constraint is just the original primal program, but with the equality constraint $Mp=p'$. Thus, for any feasible $M$ for the original primal problem, there exists a feasible $M$ for the modified primal problem, and the forward implication holds. 
\end{proof}

The following dilation lemma is an adaptation of a result by the author in~\cite{renes_work_2014}.
See Eq.\ 2.C.4 of \cite{marshall_inequalities:_2009} for a similar dilation of a doubly substochastic matrix to a doubly stochastic matrix which enables dilation of standard submajorization into standard majorization. 
\begin{lemma}
\label{lem:uglydual}
\label{lem:lowerrm}
For any $p,q\in \mathbb R_+^n$ and  $p',q'\in \mathbb R_+^{n'}$, let $z=|q'|/|q|$. Then the following are equivalent
\begin{enumerate}[(a)]
\item $(p,q)\succ (p',q')$,
\item $(p\oplus 0,q\oplus z q)\succcurlyeq (p'\oplus 0,q'\oplus z^{-1}q')$,
\item $(p\oplus 0,q\oplus z q)\succeq (p'\oplus r,q'\oplus z^{-1}q')$ for $r\propto q'$.
\end{enumerate}
\end{lemma}
\begin{proof}
We begin with the easier implications. First, (c) clearly implies (b). Condition (b) can be expressed in terms of the blocks of $M$ as
\begin{subequations}
\label{eq:Mblock}
\begin{align}
M_{11}p\geq p',&\quad M_{21}p\geq 0,\\
M_{11}q+z M_{12}q= q',&\quad M_{21}q+zM_{22}q=z^{-1} q',\\
1_{n'}^TM_{11}+1_{n'}^TM_{21}= 1_n^T,&\quad 1_{n'}^TM_{12}+1_{n'}^TM_{22}= 1_n^T.
\end{align}
\end{subequations}
Using the condition involving $M_{11}$ in each pair and positivity of $M_{12}$, $M_{21}$, $q$, and $z$, we see that $M_{11}p\geq p'$, $M_{11}q\leq  q'$ and $M_{11}$ is substochastic. Thus, (b)  implies (a). 

To establish (a)$\Rightarrow$ (b), we construct the appropriate stochastic matrix for (b)  from the substochastic matrix $F$ of (a). 
Suppose $F$ is a substochastic matrix such that $Fp\geq p'$ and $Fq\leq q'$. Define
\begin{align}
u&\defeq q'-Fq\qquad \text{and}\\
v^T&\defeq 1_{n}^T-1_{n'}^TF,
\end{align}
which are both nonnegative vectors. 
Then set 
\begin{align}
\begin{array}{lcl}
M_{11}=F, &\quad &M_{12}=\dfrac1{|q'|} u1_{n}^T\\[5mm]
M_{21}=\dfrac1{|q'|}q'v^T, &\quad& M_{22}=\dfrac1{|q'|^2} 1_{n'}^TFq\,\, q' 1_n^T.
\end{array}
\end{align} 
Now we can verify that the required conditions from \eqref{eq:Mblock} are fulfilled. 
In the first set we have $M_{11}p\geq p'$ and $M_{21}p=q'\, v^Tp/|q'|\geq 0$, since $v^T\geq 0$. 
In the second, we have equality:
\begin{align}
M_{11}q+zM_{12}q&=Fq+u=q'\quad \text{and}\\
zM_{21}q+z^2 M_{22}q 
&=\frac{1}{|q|}q'(v^T+1_{n'}^TF)q=q'.
\end{align}
The same holds for the third:  
\begin{align}
1_{n'}^TM_{11}+1_{n'}^TM_{21}
&= 1_{n'}^TF+v^T=1_{n}^T\quad \text{and}\\
1_{n'}^TM_{12}+1_{n'}^TM_{22}
&=\frac 1{|q'|}(1_{n'}^Tu+1_{n'}^TFq)1_n^T=1_n^T.
\end{align}

Finally, for (b)$\Rightarrow$(c), use the equivalence from Lemma~\ref{lem:upperrm} and repeat the above argument starting from substochastic $F$ satsifying $Fp=p'$ and $Fq\leq q'$. The necessary $r$ is simply a suitably-normalized version of $q'$: $r=q'v^Tp/|q'|$.
\end{proof}

Before proceeding to the equivalence between submajorization and ordering of Lorenz curves, it is convenient to state an additional lemma. 
To do so, we define the function $f_a:x \mapsto (x-a)$ and the associated ``angle'' function $f_a^+:x\mapsto (x-a)_+$, as well as the homogeneous two-dimensional variants $h_{u,v}:(x,y)\mapsto ux-vy$ and $h_{u,v}^+:(x,y)\mapsto (ux-vy)_+$. 
Observe that, for $u,y>0$, the two are related by a perspective transformation, namely $h_{u,v}(x,y)=uyf_{v/u}(x/y)$ and $h_{u,v}^+(x,y)=uyf_{v/u}^+(x/y)$.
\begin{lemma}[cf.\ Th.\ 15 of \cite{torgersen_comparison_1970}]
\label{lem:anglefunctions}
For any set of $\{(u_k,v_k)\in \mathbb R_+^2\}_{k=1}^n$, there exists a set $\{(a_j,b_j,c_j)\in \mathbb R_+^3\}_{j=1}^n$ such that for all $x\geq 0$
\begin{align}
\max_k\, h_{u_k,v_k}^+(x,y)\geq \sum_{j=1}^n a_j h_{b_j,c_j}^+(x,y),
\end{align} 
where equality holds when $y\geq 0$ and the inequality holds for $y<0$.

Similarly, when only the $u_k$ are required to be positive, then there exists a set $\{(a_j,b_j,c_j)\in \mathbb R_+^3\}_{j=0}^n$ such that, for all $(x,y)\in \mathbb R_+^2$, 
\begin{align}
\max_{k}  h_{u_k,v_k}(x,y) =a_0(b_0x-c_0y)+\sum_{j=1}^n a_j h_{b_j,c_j}^+(x,y).
\end{align}
\end{lemma}
\begin{proof}
For $(x,y)$ with $y>0$, we may convert each $h_{u_k,v_k}^+$ to its perspective angle function $f^+_{u_k/v_k}$. 
This reduces the problem to one dimension, the line $y=1$.
Each $h_{u_k,v_k}^+$ is increasing along this line, since $u_k>0$ is the component of the gradient in this direction. 
Moreover, with $v_k\geq 0$, the maximum is a maximum of angle functions with elbows at positive $x$ values. 
Not all intersections between the $h_{u_k,v_k}^+$ are found in the positive quadrant, but only those that do are relevant for the behavior of the maximization along the $y=1$ line. 
Hence, there are at most $n$ intersections, and it is easy to see that any maximum of angle functions is itself just a positive linear combination of angle functions. 

Converting back to $h_{u_k,v_k}^+$ gives the desired statement for $x\geq 0$ and $y>0$, and the case $y=0$ follows by continuity of $h_{u,v}^+$. 
The one-dimensional restriction misses all the intersections that take place outside the positive quadrant, so only the inequality holds for points in this region.

The only difference in the case of maximizing $h_{u_k,v_k}$ instead of $h_{u_k,v_k}^+$ is that the linear functions along the line $y=1$ no longer necessarily take the value 0 at some positive $x$. 
Thus, to recreate the maximization it is necessary to add a linear function and a constant function to the positive linear combination of angle functions. 
\end{proof}

\begin{lemma}
\label{lem:directroute}
For any $p,q\in \mathbb R_+^n$ and $p',q'\in \mathbb R_+^{n'}$, the following are equivalent:
\begin{enumerate}[(a)]
\item  $(p,q)\succ (p',q')$
\item $\beta_x(p,q)\leq \beta_x(p',q')$ for all $x\in \mathbb R$. 
\end{enumerate}
Furthermore, if $|q'|=|q|$, then these two are also equivalent to $(p,q)\succcurlyeq (p',q')$.
\end{lemma}
\begin{proof}
To start, assume that $(p,q)\succ (p',q')$. 
For any $x\in [0,|p'|]$, let $t'$ be the optimal test in $\beta_x(p',q')$, so that $t'\cdot p'=x$ and $t'\cdot q'=\beta_x(p',q')$. 
The vector $t=M^T t'$ satisfies $0\leq t\leq 1_n$ by the stochasticity of $M$, and $t\cdot p=t'\cdot Mp\geq t'\cdot p'=x$. 
Therefore, $t$ is feasible for $\beta_x(p,q)$, and we immediately have $\beta_x(p,q)\leq t\cdot q=t'\cdot Mq \leq t'\cdot q'=\beta_x(p',q')$.

For the implication from (b) to (a), we show that assuming both $(p,q)\nsucc (p',q')$ and ordered Lorenz curves leads to a contradiction. 
Consider the feasibility linear program for relative submajorization, particularly the dual version given in \eqref{eq:feasLPdual}.
Infeasibility of submajorization means there exists a feasible set of dual variables $u,v,w$ such that the objective function exceeds zero.
Observe that, for any feasible $u$ and $v$, the optimal $w$ has components $w_j=\max_k(p_ju_k-q_jv_k)_+=\max_k h_{u_k,v_k}^+(p_j,q_j)$. 
Since $u\cdot p'-v\cdot q'\leq \sum_j \max_k h_{u_k,v_k}^+(p'_j,q'_j)$, for any $u$ and $v$ such that the objective function exceeds zero, we have
\begin{align}
\sum_{j=1}^n \max_k\, h_{u_k,v_k}^+(p_j,q_j) &<\sum_{j=1}^{n'} \max_k \,h_{u_k,v_k}^+(p'_j,q'_j).
\end{align}
By Lemma~\ref{lem:anglefunctions}, we can convert the maximization into a positive linear combination of $h_{b_k,c_k}^+$. 
Factoring out $c_k$ and defining $\mu_k=b_k/c_k$, we see that $(p,q)\nsucc (p',q')$ implies 
\begin{align}
\label{eq:dualconclusion}
\sum_k a_k c_k\sum_j (\mu_k p_j-q_j)_+<\sum_k a_kc_k \sum_j (\mu_kp'_j-q'_j)_+.
\end{align}

Now consider the dual formulation of $\beta_x$, given in \eqref{eq:betadual}. 
The optimal $s$ is clearly $s=(\mu p-q)_+$, which gives $\beta_x(p,q)=\max\{\mu x-1\cdot (\mu p-q)_+:\mu\geq 0\}$. 
Therefore, $\mu\mapsto 1\cdot(\mu p-q)_+$ is the Legendre transform of $x\mapsto \beta_x(p,q)$. 
As the Legendre transform is order-reversing, $\beta_x(p,q)\leq \beta_x(p'q')$ for all $x\in\mathbb R$ implies, for all $\mu \geq 0$,  
\begin{align}
\label{eq:betaconclusion}
\sum_j (\mu p_j-q_j)_+\geq \sum_j (\mu p'_j-q'_j)_+.
\end{align}
This contradicts \eqref{eq:dualconclusion}, and hence the Lorenz curves cannot be ordered if $(p,q)\nsucc (p',q')$.

Finally, to establish the equivalence to $(p,q)\succcurlyeq (p',q')$ when $|q'|=|q|$, we need only show this statement is implied by (b), since it implies (a).
Again we argue via contradiction, and the argument is similar to that above. 
In this case, the dual of the feasibility linear program only requires $u\geq 0$, but $v$ and $w$ are unconstrained. 
The optimal $w$ now has components $w_j=\max_k(p_ju_k-q_jv_k)=\max_k h_{u_k,v_k}(p_j,q_j)$.
Since $u\cdot p'-v\cdot q'\leq \sum_j \max_k h_{u_k,v_k}(p_j',q_j')$, for any $u$ and $v$ such that the objective function exceeds zero, we now have
\begin{align}
\sum_{j=1}^n \max_k\, h_{u_k,v_k}(p_j,q_j) &<\sum_{j=1}^{n'} \max_k \,h_{u_k,v_k}(p'_j,q'_j).
\end{align}
Using the second part of Lemma~\ref{lem:anglefunctions}, we can again convert this to a positive linear combination of $h_{b_k,c_k}^+$, as in \eqref{eq:dualconclusion}, plus additional linear terms: $\sum_j a_0(b_0 p_j-c_0q_j)$ on the left and $\sum_j a_0(b_0 p'_j-c_0q'_j)$ on the right. 
But the former is just $a_0(b_0|p|-c_0 |q|)$ and the latter $a_0(b_0|p'|-c_0 |q'|)$. 
The $|q|$ and $|q'|$ terms are identical, so we are left with
\begin{align}
\label{eq:dualconclusion2}
a_0b_0 |p|+\sum_{k=1}^n \sum_j a_k c_k (\mu_k p_j-q_j)_+< a_0b_0 |p'|+ \sum_k a_kc_k \sum_j (\mu_kp'_j-q'_j)_+.
\end{align}
But this condition certainly cannot hold if (b) does, using \eqref{eq:betaconclusion} and the fact that the $\beta_x$ ordering implies $|p|\geq |p'|$ (which is more easily seen using $\alpha_y$). 
\end{proof}

Finally, we formalize the statement that one only need check for ordering of the Lorenz curves at the elbows of ostensibly submajorized pair.
\begin{lemma}
\label{lem:elbows}
For any $p,q\in \mathbb R_+^n$ and $p',q'\in \mathbb R_+^{n'}$, the condition $\beta_x(p,q)\leq \beta_x(p',q')$ for all $x\in [0,|p'|]$ holds if and only if it holds for all $x\in T_x^\star(p',q')$ (at the elbows of $(p',q')$). 
\end{lemma}
\begin{proof}
Clearly the ordering of the Lorenz curves at the elbows of $(p',q')$ is necessary.
Sufficiency follows from the fact that the function $x\mapsto \beta_x(p',q')$ interpolates linearly between $\beta_k^\star(p',q')$ and $\beta_{k+1}^\star(p',q')$ for the largest $k$ such that $\alpha_k^\star(p',q')<x$, while the graph of $x\mapsto \beta_x(p,q)$ lies below these two endpoints and is convex.
\end{proof}

As mentioned in the main text, it is possible to construct the necessary stochastic matrix $M$ such that $Mp= p'$ and $Mq\leq q'$  from the Lorenz curves of $(p,q)$ and $(p',q')$ themselves. 
Here we sketch the argument, which is based on \S5 of \cite{torgersen_stochastic_1991}.
First, suppose that we find a sequence of tests $0=t_0\leq t_1\leq \dots t_k\leq t_{n'}\leq 1$ such that $t_k\cdot p=\sum_{j=1}^k p'_{\pi(j)}$ and $t_k\cdot q\leq \sum_{j=1}^k q'_{\pi(j)}$, where again $\pi$ is a permutation as in Lemma~\ref{lem:Tcornerpoints}. 
Then observe the differences between subsequent tests form the rows of a substochastic matrix such that $Mp=p'$ and $Mq\leq q'$. 
Finally, it is straightforward to construct an appropriate sequence of tests by interpolating between the optimal tests corresponding to the elbows of the Lorenz curve of $(p,q)$. 
To start, consider the line joining the origin and the first elbow of the Lorenz curve of $(p',q')$. 
This line intersects the Lorenz curve of $(p,q)$ at some point, corresponding to a test $\hat t_1$.
Interpolating between $\hat t_1$ and the zero test, we can construct a sequence of tests which yields the first segment of the $(p',q')$ Lorenz curve. 
By similarly extending the next segment, we find the appropriate test $\hat t_2$, and continue. 
Should the extended line of a particular segment fail to intersect the $(p,q)$ Lorenz curve (because the slope is too small), we can simply take the associated test $\hat t=1$.

A slight modification of the proof of Lemma~\ref{lem:directroute} gives the \hypertarget{proof:th2}{proof} of Theorem~\ref{thm:approxLorenzbeta}.

\begin{proof}[Proof of Theorem~\ref{thm:approxLorenzbeta}]

The easier implication is from (a) to (b), so we begin here. 
Using \eqref{eq:approxLorenzdef} in the definition of $\beta_x$, we have 
$\beta_x(p,q)\leq \beta_x(p'-a,q'+b)$ for all $x\in [0,|p'|-\eps]$.
Let $t$ be the optimal test for $\beta_{x+\eps}(p',q')$; then $t\cdot p'=x+\eps\geq x+t\cdot a$ as $t\leq 1$. 
Therefore $t$ is feasible in $\beta_x(p'-a,q'+b)$, implying $\beta_x(p'-a,q'+b)\leq \beta_{x+\eps}(p',q')+\eta$.

To show (b) implies (a), we show that assuming (b) and not-(a) leads to contradiction.
First consider (b). 
Taking the Legendre transform gives the reverse inequality
\begin{align}
\sup_x \mu x-\beta_x(p,q)\geq \sup_x \mu x-\beta_{x+\eps}(p',q')-\eta,
\end{align}
valid for all $\mu\in \mathbb R$. 
The lefthand side can be expressed as 
\begin{align}
\sup_x \mu x-\beta_x(p,q)=\sup_{t:0\leq t\leq 1} \mu t\cdot p-t\cdot q,
\end{align}
which makes it clear that the solution is just $1\cdot (\mu p-q)_+$.
Proceeding similarly on the righthand side leads to
\begin{align}
\label{eq:betaorderimplication}
1\cdot (\mu p-q)_+\geq 1\cdot (\mu p'-q')_+-\mu\eps-\eta.
\end{align}

For the negation of (a), we appeal to the formulation of approximate submajorization as a dual pair of linear programs.
\renewcommand{\arraystretch}{1.2}
\begin{align}
\label{eq:feasapproxLPprimal}
&\begin{array}{r@{\,\,}rl}
    c_{\eps,\eta}=& \underset{M}{\rm infimum} &  0\\
     &\text{\rm subject to} & Mp\geq p'-a\\
     && Mq\leq q'+b,\\
     && 1^T_mM\leq1^T_n\\
     && 1^T_m a\leq \eps\\
     && 1^T_m b\leq \eta\\
     &&a,b,M\geq 0,
\end{array}\\
\label{eq:feasapproxLPdual}
&\begin{array}{r@{\,\,}rl}
     c_{\eps,\eta}=& \underset{u,v,w,\mu,\nu}{\rm supremum} &  u^Tp'-v^Tq'-1^T_nw-\mu\eps-\nu\eta\\
     &\text{\rm subject to} & - pu^T+qv^T+w 1_{n'}^T\geq 0,\\
     &&u\leq \mu 1_{n'}\\
     &&v\leq \nu 1_{n'}\\
     &&\mu,\nu,u,v,w\geq 0.
\end{array}
\end{align}
\renewcommand{\arraystretch}{1}
When $(p,q)\nsucc_{\eps,\eta}(p',q')$, then in the dual formulation there are feasible points for which the objective exceeds zero. 
Again the optimal $w$ has components $w_j=\max_k(p_ju_k-q_jv_k)_+=\max_k h_{u_k,v_k}(p_j,q_j)$, while the optimal $\mu$ and $\nu$ are clearly $\max_k u_k$ and $\max_k v_k$, respectively. 
Using $h_{u,v}$, we can rewrite the infeasibility condition in terms of the objective function as 
\begin{align}
\sum_j \max_k h_{u_k,v_k}(p'_j,q'_j) - \sum_j \max_k  h_{u_k,v_k}(p_j,q_j) - \max_k  h_{u_k,v_k}(\eps,-\eta)>0.
\end{align}
Applying Lemma~\ref{lem:anglefunctions} and interchanging the summation order gives
\begin{align}
\sum_k a_k c_k \left(\sum_j  (\mu_k p'_j-q'_j)_+ -  \sum_j (\mu_k p_j-q_j)_+ - (\mu_k\eps+\eta)_+\right)>0,
\end{align}
for $\mu_k=b_k/c_k\geq 0$. 
However, each term contradicts \eqref{eq:betaorderimplication}, meaning not-(a) and (b) cannot both be true. 
\end{proof}

\begin{proof}[\hypertarget{proof:prop5}{Proof} of Prop.~\ref{prop:approxRT}, Eq.~\eqref{eq:optimaleta}]
We give the proof for a battery with three levels, which serves to illustrate the general idea.
Let the battery levels have energies $E$, $E+W$, and $E+W'$, and set $r\otimes a(E)=r\oplus 0\oplus 0$, $r\otimes a(E+W)=0\oplus r\oplus 0$, and $r\otimes a(E+W')=0\oplus 0\oplus r$.
Furthermore, let $z=e^{-\beta W}$ and $z'=e^{-\beta W'}$ as usual.

Now observe that by using the dual formulation of $\delta$, we may formulate $\eta^\star_z(R\to R')$ feasibility as a linear program:
\renewcommand{\arraystretch}{1.2}
\begin{align}
\begin{array}{r@{\,\,}rl}
     \eta^\star_z(R\to R')= & \underset{d,M}{\rm minimize} &  1_{3n'}^T d \\
     &\text{\rm subject to} & M(r\oplus 0\oplus 0) \geq 0\oplus r'\oplus 0,\\
     && Z_Be^{\beta E} d\geq M(g\oplus zg\oplus z' g)-g'\oplus z g'\oplus z'g',\\
     && 1^T_{3n'}M\leq  1^T_{3n},\\
     &&d\in \mathbb R^{3n'}_+, M\in \mathbb R^{3n'\times 3n}_+.
\end{array}
\end{align}
The dual form is 
\begin{align}
\begin{array}{r@{\,\,}rl}
     \eta^\star(R\to R')= & \underset{u,v,w}{\rm maximize} &  u^T (0\oplus r'\oplus 0)-v^T(g'\oplus z g'\oplus z'g')-1_{3n}^Tw\\
     &\text{\rm subject to} & 
     Z_Be^{\beta E}v\leq 1^T_{3n'},\\
     && (r\oplus 0\oplus 0)u^T-(g\oplus zg\oplus z' g)v^T-w1_{3n'}^T\leq 0,\\
     &&u,v\in \mathbb R^{3n'}_+,w\in \mathbb R^{3n}_+.
\end{array}
\end{align}
\renewcommand{\arraystretch}{1}
Strong duality holds since all dual variables equal to zero is a feasible choice. 
Now let $u=u_1\oplus u_2\oplus u_3$ and the same for $v$ and $w$. 
In terms of these variables, the matrix constraint becomes nine individual constraints: $ru_1^T-gv_1^T-w_1 1_n^T\leq 0$ and so on. 
However, six of these are superfluous, since those which do not involve any $u$ component are satisfied for any permissible choice of $v$ and $w$. 
This leaves 
\begin{align}
\renewcommand{\arraystretch}{1.2}
\begin{array}{r@{\,\,}rl}
     \eta^\star(R\to R')= & \underset{u,v,w}{\rm maximize} &  u_2^T r'-v^T_1g'-zv_2^Tg'-z'v_3^Tg'-1_{n}^Tw_1-1_n^Tw_2-1_n^Tw_3\\
     &\text{\rm subject to} & 
     Z_Be^{\beta E}v\leq 1^T_{3n'},\\
     && ru_1^T-gv_1^T-w_1 1_n^T\leq 0,\\
     && ru_2^T-gv_2^T-w_1 1_n^T\leq 0,\\
     && ru_3^T-gv_3^T-w_1 1_n^T\leq 0,\\
     &&u,v,w\geq 0,\\
     &&u,v\in \mathbb R^{3n'},w\in \mathbb R^{3n}.
\end{array}
\end{align}
\renewcommand{\arraystretch}{1}

Now, for any feasible set of variables, the same set with $u_1=u_3=0$ is also feasible and has the same value of the objective function. 
Furthermore, setting $v_1=v_3=w_2=w_3=0$ in any feasible set produces another, but with possibly larger objective. 
Therefore, we have
\begin{align}
&\begin{array}{r@{\,\,}rl}
     \eta^\star_z(R\to R')=& \underset{u,v,w}{\rm maximize} &  u^T r'-zv^Tg'-1_{n}^Tw_1\\
     &\text{\rm subject to} & 
     Z_Be^{\beta E}v\leq 1^T_{n'},\\
     && ru^T-gv^T-w 1_{n'}^T\leq 0,\\
     &&u,v\in \mathbb R^{n'}_+,w\in \mathbb R^{n}_+.
\end{array}\label{eq:reducedetadual}\\
&\begin{array}{r@{\,\,}rl}
     \phantom{\eta^\star(R\to R')}= & \underset{d,M}{\rm minimize} &  1_{n'}^T d \\
     &\text{\rm subject to} & Mr \geq  r',\\
     && Z_Be^{\beta E} d\geq Mg-zg',\\
     && 1^T_{n'}M\leq  1^T_{n},\\
     &&d\in \mathbb R^{n'}_+, M\in \mathbb R^{n'\times n}_+.
\end{array}
\label{eq:reducedetaprimal}
\end{align}
Subsuming $Z_Be^{-\beta E}$ into $d$ in the latter expression, we see that it is linear program for the smallest $\delta(s',zg')\frac{e^{-\beta E}}{Z_B}$ such that $(r,g)\succ (r',s')$ for some $s'\in \mathbb R_+^{n'}$. 
Since the same proof carries over to a battery with an arbitrary number of levels, this gives the latter two statements of the proposition. 
\end{proof}

\begin{proof}[\hypertarget{proof:prop11}{Proof} of Prop.~\ref{prop:etazbound}]
Start with $(r,g)\succ (r',z^\star g')$.
This entails the existence of a substochastic $M$ such that $Mr\geq r'$ and $g'':=Mg\leq z^\star g'$; indeed, there must be some $k$ such that $g''_k =(Mg)_k=z^\star g_k$, else a smaller $z$ would be feasible.
We may use this $M$ in the primal form of $\hat \eta_{z}^\star$, \eqref{eq:reducedetaprimal}.
The optimal $d$ given this (or any) choice of $M$ is clearly $(g''-zg')_+$. 
This yields 
\begin{align}
\hat\eta^\star_{z} &\leq \sum_{k=1}^{n'}(g''_k-zg'_k)_+\\
&= \sum_{k=1}^{n'} g''_k(1-z'\frac {g'_k}{g''_k})_+\\
&\leq \max_k\, (1-z\frac {g'_k}{g''_k})\\
&=1-\frac{z}{\max_k \frac{g''_k}{g'_k}}\\
&=1-\frac {z}{z^\star},
\end{align}
which is \eqref{eq:etazgeq1}.
\end{proof}

\begin{lemma}[Statement 3~\cite{winter_stronger_2012}]
\label{lem:petz}
Given any $p\in S^n$, $q\in \mathbb R_+^n$, and any $n'\times n$ stochastic matrix $T$, there exists an $n\times n'$ stochastic matrix $\hat T$ depending only on $T$ and $q$, such that $\hat TTq=q$ and  
\begin{align}
D(p,q)-D(Tp,Tq)\geq D(p,\hat TTp).
\end{align}
\end{lemma}
\begin{proof}
The map $\hat T$ is the so-called Petz map~\cite{petz_sufficiency_1988},\cite[Th.\ 3]{hayden_structure_2004} (see also the independent definition in Eq.\ 9 of \cite{barnum_reversing_2002}), which is constructed from the stochastic map $T$ and the vector $g$. 
In particular, if we define $u\boxslash v$ to be the component-wise ratio of the vectors $u$ and $v$ and $u\boxtimes v$ to be their component-wise product, then the Petz map $\hat T:\mathbb R^{n'}\to \mathbb R^n$ has the action  
\begin{align}
\hat T:v\rightarrow g\boxtimes T^T\left( v\boxslash {Tg}\right).
\end{align}
Since $T$ is stochastic, it follows at once that $\hat T g'=g$.
It is not immediately obvious that $\hat T$ is itself stochastic, but by using components,
\begin{align}
(\hat T v)_k=g_k\sum_{j} T_{jk}\frac{v_j}{\sum_{k'}T_{jk'}g_{k'}},
\end{align}
this follows by direct calculation. 

Observe that the quantity $D(r,\hat T Tr)$ involves 
\begin{align}
(\log \hat TTr)_k=\log g_k+\log\left(\sum_j T_{jk}\frac{r'_j}{g'_j}\right).
\end{align}
Since $\{T_{jk}\}_{j=1}^{n'}$ is a probability distribution for each $k$, concavity of the logarithm and Jensen's inequality implies 
\begin{align}
(\log \hat TTr)_k\geq \log g_k+\sum_j T_{jk}\log\left(\frac{s'_j}{g'_j}\right).
\end{align}
Thus, for each term in $D(r,\hat T Tr)$ we have
\begin{align}
r_k \log r_k - r_k \log ( \hat TTr)_k \leq r_k \log \frac {r_k}{g_k}-r_k\sum_j T_{jk}\log\left(\frac{s'_j}{g'_j}\right).
\end{align}
Summing over $k$ gives the desired result.
\end{proof}

\section{Recovering asymptotic results}
\label{sec:asymptotics}
Using \eqref{eq:lambdaofz} and \eqref{eq:zoflambda} we can very easily rederive many results on transitions involving asymptotically many resources. 
For instance, in the limit of large $N$, the work gain of the transformation $R^n\to R'^N$ is governed by the relative entropy difference $D(r,g)-D(r',g')$, as shown in~\cite{brandao_resource_2013}. 
This follows immediately from \eqref{eq:zoflambda}, as we have
\begin{align}
\lim_{N\to \infty}\,-\frac 1N \log z^\star(R^N\to R'^N)
&=\lim_{N\to \infty}\,-\frac 1N \log \max_{x\in (0,1]}\frac{\beta_{x}(r^{\otimes N},g^{\otimes N})} {\beta_{x}(r'^{\otimes N},g'^{\otimes N})}\\
&=\lim_{N\to \infty}\,-\frac 1N \log\frac{\beta_{x^\star}(r^{\otimes N},g^{\otimes N})} {\beta_{x^\star}(r'^{\otimes N},g'^{\otimes N})}\\
&=\lim_{N\to \infty}\,-\frac 1N \left(\log {\beta_{x^\star}(r^{\otimes N},g^{\otimes N})} -\log{\beta_{x^\star}(r'^{\otimes N},g'^{\otimes N})}\right)\\
&=D(r,g)-D(r',g').
\end{align}
Here, $x^\star$ is the optimal value of $x$ and in the last line we use Stein's lemma, which holds for any $x^\star\in (0,1)$~\cite[Th.\ 4.4.4]{blahut_principles_1987}.
When $x^\star=1$, we must take a slight detour.
First, the upper bound $\lim_{N\to \infty}\,-\frac 1N \log z^\star(R^N\to R'^N)\leq D(r,g)-D(r',g')$ holds by lower bounding the maximum ratio, using an $x<1$.  
Second, the maximum ratio is continuous at $x=1$, since it is the ratio of two linear functions which intersect at $x=1$. 
Thus, by choosing $x$ close enough to $1$, we can ensure that the maximum ratio is upper bounded by, say, twice the value at $x$. 
This gives a lower bound of precisely the same form, as the additional $\log 2$ term will vanish in the limit.

Choosing $R'$ to be the trivial resource $r'=g'=1$, we recover the asymptotic work value of $R$, namely $D(r,g)$, since $D(1,1)=0$. 
Similarly, setting $R$ to be the trivial resource, the asymptotic work cost of $R'$ is $D(r',g')$. 
Chaining these two together, it is apparent that the rate of reversible interconversion of $R$ to $R'$ must be the ratio $D(r,g)/D(r',g')$. 
To perform the transformation, we simply extract all the available work from $R$ and use it to create instances of $R'$. 
Since the asymptotic work cost and gain are equal, the transformation is reversible, at least in the sense that only a sublinear amount of work (zero rate) needs to be supplied.

Note that the transformation so constructed is slightly different to that of~\cite{brandao_resource_2013}, though the approach of splitting the transformation into two steps is similar.
There, each step is performed approximately, in that the output of the first step is not precisely the input considered in the second step, but in a way that ensures the approximation error goes to zero in the limit of large $N$. 
Here no such approximations are employed; instead the precise sublinear amount of work required may be different.

We cannot usefully apply Stein's lemma in \eqref{eq:lambdaofz} to obtain an asymptotic limit for $\lambda^\star(R\to R')$, since $\alpha_y$ is not a type-II error, but one minus a type-II error, as evidenced by \eqref{eq:alphabeta}. 
Another way to see this is that $\alpha_y$ involves a maximization, not a minimization as in $\beta_x$. 
However, we can use Stein's lemma in \eqref{eq:lambdaofz} to rederive Theorem 8 from \cite{janzing_thermodynamic_2000} that the reverse relative entropy $D(g,r)$ determines the asymptotics of using $R^N$ to cool (or heat) a two-level system. 
Similarly, $D(g,r)$ quantifies the difficulty of creating erased bits, a two-level resource $R_b=(b(\texttt{0}),g_b)$.

Beginning with the latter example, $\lambda^\star(R^N\to R_b)$ is the largest probability of creating a perfectly erased bit, or equivalently, one minus the approximation error $\eps^\star(R^N\to R_b)$ in doing so. 
Since $\alpha_y((1,0),g_b)$ has only one elbow, at $(\frac 12,1)$, \eqref{eq:lambdaofz} gives
\begin{align}
\lambda^\star(R^N\to R_b)=\alpha_{\frac 12}(r^{\otimes N},g^{\otimes N}).
\end{align}
In view of \eqref{eq:alphabeta} we have
\begin{align}
\eps^\star(R^N\to R_b)=\beta_{\frac 12}(g^{\otimes N},r^{\otimes N}),
\end{align}
and thus, by Stein's lemma
\begin{align}
\lim_{N\to \infty}-\frac 1N\log \eps^\star(R^N\to R_b)=D(g,r).
\end{align}
Thus, the error in creating a perfectly erased bit decreases exponentially with $N$.

Continuing the with former example of a two-level system with energy gap $E$, let $\beta_{N,\beta_0}$ be the lowest temperature that can be created by using $N$ instances of a resource $R=(r,g)$ and thermal operations at background temperature $\beta_0$. 
Since a two-level system with distribution $(\lambda,1-\lambda)$ and energy gap $E$ is a canonical state with temperature determined by $e^{-\beta E}=(1-\lambda)/\lambda$ (where the second level has the higher energy), we are just interested in the largest $\lambda$ such that $(r^N,g^N)\succ ((\lambda,1-\lambda),g')$,  where $g'$ is the canoncial state at temperature $\beta_0$. 
But this is just $\lambda^\star(R^N\to R')$ with $R'=((1,0),g')$. 
Replaying the above argument, the elbow in $\alpha_y((1,0),g')$ is at $(g'_1,1)$, so \eqref{eq:lambdaofz} and \eqref{eq:alphabeta} give 
\begin{align}
1-\lambda^\star(R^N\to R')=\beta_{1-g'_1}(g^{\otimes N},r^{\otimes N}).
\end{align}
Then, the limit of large $N$ yields
\begin{align}
\lim_{N\to \infty} \frac{1}{N}\beta_{N,\beta_0}E
&=\lim_{N\to \infty} -\frac{1}N \log \left(1-\lambda^\star(R^N\to R')\right)+\lim_{N\to \infty} \frac{1}N \log\lambda^\star(R^N\to R')\\
&=\lim_{N\to \infty} -\frac{1}N \log \beta_{1-g'_1}(g^{\otimes N},r^{\otimes N})\\
&=D(g,r).
\end{align}
In the second equality we use the fact that $\lambda^\star(R^N\to R')$ tends to one with increasing $N$.
Thus, the lowest temperature decreases in proportion to $N$ at a rate given by the reversed relative entropy. 

We could just as well set $R'=((0,1),g')$, so that we are trying to create a two-level system with a temperature inversion. 
Denoting by $\hat\beta_{N,\beta_0}$ the highest temperature inversion, repeating the argument above yields
\begin{align}
\lim_{N\to \infty} \frac{1}{N}\hat \beta_{N,\beta_0}E=-D(g,r).
\end{align}

\section{Comparison with the work extraction game of Egloff \emph{et al.}~\cite{egloff_measure_2015}}
\label{sec:egloff}
There are two major differences between the present approach to modelling thermodynamic operations and the approach by Egloff~\emph{et al.}: the set of allowed operations and the definition of work. 
Nevertheless, we shall see that the operations allowed in their model are a subset of those in the resource theory model as defined in \S\ref{sec:thermalops}. 
From our perspective, we can regard \cite{egloff_measure_2015} as showing that a restricted set of thermal operations suffice to achieve the optimal $\lambda$ and $z$ values in \eqref{eq:lambdaofz} and \eqref{eq:zoflambda}.

In their model, which they term the ``work extraction game'', there are two kinds of allowed operations.
In the first, the resource system is allowed to thermalize with the bath. 
In the second, the Hamiltonian of the resource system is changed, with all energy shifts being taken up by (energy-conserving) interaction with the battery. 
Both of these steps are thermal operations as we have defined them, the former manifestly so.
Janzing \emph{et al.}~\cite{janzing_thermodynamic_2000} showed how to model a changing Hamiltonian in the resource theory framework, simply by introducing a new system in the Gibbs state of the new Hamiltonian and essentially just swapping the state of the resource system onto the new system, taking up any energy changes via the battery as needed. 
Therefore, the operations allowed in the work extraction game are a subset of thermal operations; in particular, the resource system is not allowed to interact with the battery and the bath at the same time. 

The notion of work in \cite{egloff_measure_2015} is subtly different from ours, in two ways. 
The first difference lies in the probablistic or deterministic nature of extracted work. 
We say an amount $W$ of work is extracted when the energy level of the battery is increased by precisely $W$. 
In~\cite{egloff_measure_2015}, by constrast, the energy of the battery need only increase by \emph{at least} $W$. 
Their notion of work contains a probabilistic element from the outset, as the battery is not necessarily in a deterministic state at the end of the process.\footnote{Allowing the battery to accept entropy might cause one to be suspicious of the ability of their work extraction game to capture features of classical thermodynamics. However, their results show that this worry is unfounded, and more conceptually, the battery cannot be used as an entropy sink to take energy from the bath without causing the probability of work extraction to become very small.}
Nonetheless, it is possible to convert this probabilistic battery to a deterministic one by thermal operations.  
One might worry that this could lead to a violation of the second law, since we essentially want to remove entropy from the battery without incurring any work cost.  
However, as we shall now see, the conversion does not remove entropy in the sense of resetting an arbitrary state of certain degrees of freedom so much as simply throwing these entropy-bearing degrees of freedom away.

The procedure is simple. 
Imagine that the battery actually consists of the given battery, denoted by $s$, plus an additional work bit denoted by $b$.
Since the battery is anyway meant to have an infinite spectrum, this poses no particular formal difficulty.
The two levels of the work bit have energy zero and $W$, and we imagine that the process begins and ends with the work bit in the zero energy state. 
The other part of the battery begins in its zero energy state and ends in some state with energy at least $W$. 
By an energy-preserving transformation we may map all the states $a(E)\otimes b(0)$ to $a(E-W)\otimes b(W)$.
Discarding $a$, we now have an overall process which deterministically takes the work bit from energy zero to $W$. 
Put differently, we are just utilizing some of the (formally infinite) degrees of freedom of the battery to deterministically store the extracted work, while putting all the entropy into different degrees of freedom, and then discarding the latter. 

The second subtle difference in the definition of work extraction is in the treatment of the output state of the resource and battery in \eqref{eq:majprobworktrans}. 
Here we do not commit to any particular $s'$, meaning when the transformation is only probabilistically possible, the probability could come from not being able to create the desired resource state or not being able to move the energy level of the battery the required amount, or some combination of both. 
In contrast, \cite{egloff_measure_2015} demands that the desired output state of the resource be created perfectly, corresponding to an output state of the form $r'\otimes s_B'$, for some distribution $s_B'$ of the battery's energy levels having weight at least $\lambda$ on energies not smaller than $W$. 

By the above discussion, we could just as well require $s_B'$ to have weight $\lambda$ on energy level $W$ itself. 
Then observe that we still end up with the submajorization condition $(r,g)\succ (\lambda r',z g')$.
The reason is that, while \eqref{eq:majprobworktrans} has a more specific form, it still leads to \eqref{eq:convergecondition}, and thus the submajorization condition. 
We have not shown that the converse still holds in this more specific case, i.e.\ that anytime the submajorization condition holds, there is a thermal operation which will produce this particular kind of state at the output. 
But we can view the results of \cite{egloff_measure_2015} as proving precisely this claim.

\end{document}